\documentclass[11pt]{article}

\usepackage{amsfonts}
\usepackage{amsmath}
\usepackage{amssymb}
\usepackage{amstext}
\usepackage{amsthm}
\usepackage{fullpage}
\usepackage{color}
\usepackage{eucal}
\usepackage{graphicx}
\usepackage{dsfont}
\usepackage{subfigure}

\newcommand{\identity}{\raisebox{-7pt}{\includegraphics[height=18pt]{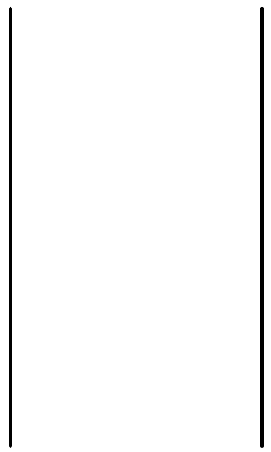}}}
\newcommand{\exchange}{\raisebox{-7pt}{\includegraphics[height=18pt]{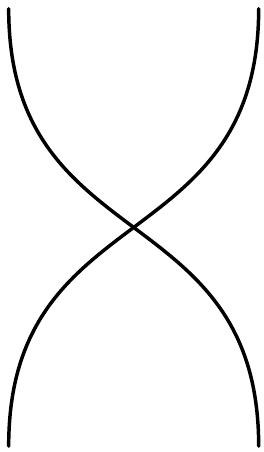}}}

\newcommand{\cupcap}{\raisebox{-7pt}{\includegraphics[height=18pt]{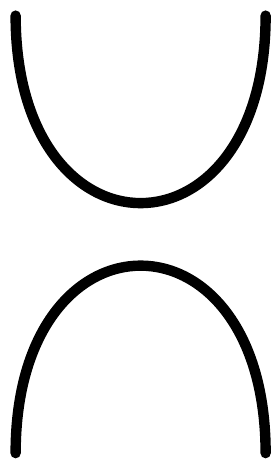}}}
\newcommand{\cupcapnest}{\raisebox{-7pt}{\includegraphics[height=18pt]{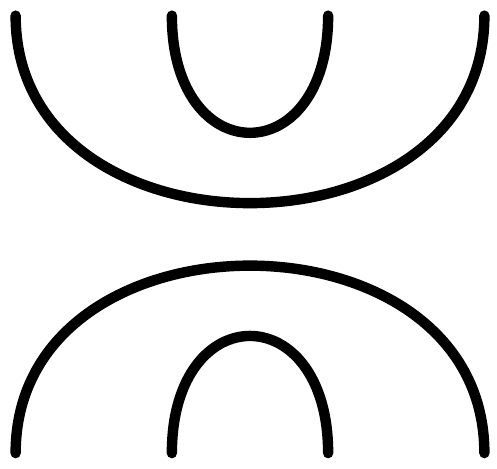}}}
\newcommand{\cupcapcross}{\raisebox{-7pt}{\includegraphics[height=18pt]{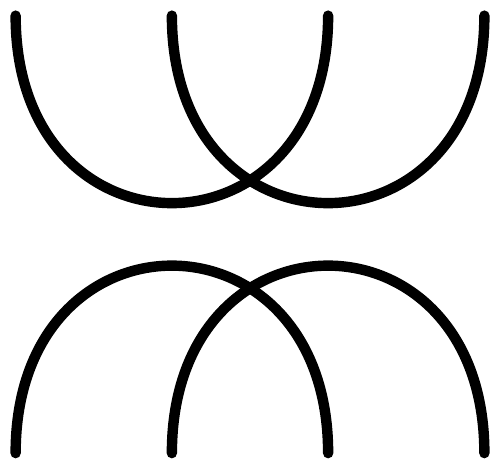}}}
\newcommand{\cupcapmixtop}{\raisebox{-7pt}{\includegraphics[height=18pt]{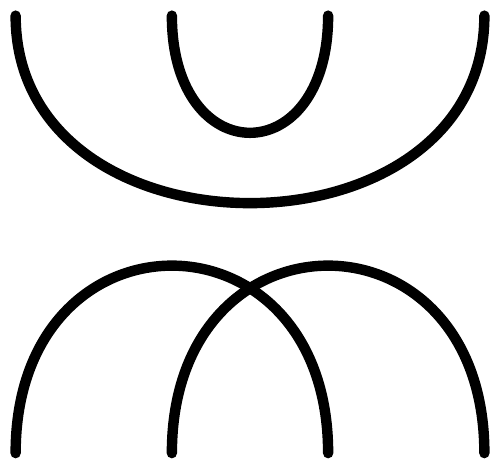}}}
\newcommand{\cupcapmixbottom}{\raisebox{-7pt}{\includegraphics[height=18pt]{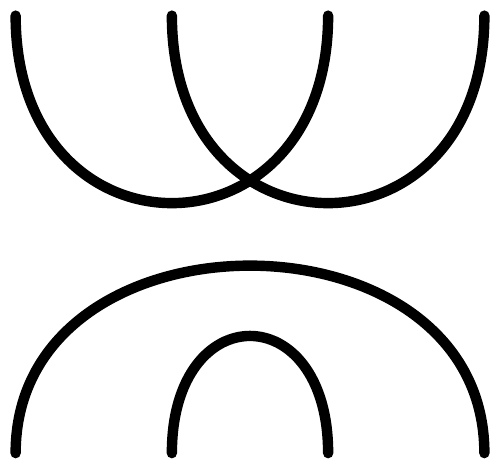}}}

\newcommand{\Xsym}{X_s}
\newcommand{\Xunsym}{X}
\newcommand{\Xfrobsym}{X_{{\rm Frob},s}}
\newcommand{\Xfrobunsym}{X_{\rm Frob}}

\newcommand{\inner}[2]{\left\langle #1, #2 \right\rangle}

\newtheorem{theorem}{Theorem}
\newtheorem*{theorem*}{Theorem}
\newtheorem{lemma}[theorem]{Lemma}
\newtheorem{conjecture}{Conjecture}

\DeclareMathOperator{\Exp}{\mathbb{E}}
\DeclareMathOperator{\Var}{\textsf{Var}}

\DeclareMathOperator{\tr}{\textbf{tr}}

\DeclareMathOperator{\sdet}{\textrm{sdet}}
\DeclareMathOperator{\perm}{\textrm{perm}}

\renewcommand{\vec}[1]{\mathbf{#1}}

\newcommand{\abs}[1]{\left| #1 \right|}
\newcommand{\norm}[1]{\left\| #1 \right\|}
\newcommand{\e}{{\rm e}}
\newcommand{\eps}{\epsilon}
\newcommand{\poly}{{\rm poly}}

\newcommand{\sym}{+}
\newcommand{\asym}{-}

\newcommand{\one}{{\mathds 1}}
\renewcommand{\vec}[1]{\mathbf{#1}}

\newcommand{\remove}[1]{}

\newcommand{\C}{\mathbb{C}}
\newcommand{\R}{\mathbb{R}}

\newcommand{\U}{\textsf{U}}
\newcommand{\GL}{\textsf{GL}}
\newcommand{\alg}{\mathcal{A}}

\title{Approximating the Permanent via Nonabelian Determinants}
\author{Cristopher Moore\thanks{\textsf{moore@santafe.edu}, Department of Computer Science, University of New Mexico and Santa Fe Institute} \and Alexander Russell\thanks{\textsf{acr@cse.uconn.edu}, Department of Computer Science and Engineering, University of Connecticut}}

\begin{document}
\maketitle

\begin{abstract}
Since the celebrated work of Jerrum, Sinclair, and Vigoda, we have known that the permanent of a $\{0,1\}$ matrix can be approximated in randomized polynomial time by using a rapidly mixing Markov chain to sample perfect matchings of a bipartite graph.  A separate strand of the literature has pursued the possibility of an alternate, \emph{algebraic} polynomial-time approximation scheme.  These schemes work by replacing each 1 with a random element of an algebra $\alg$, and considering the determinant of the resulting matrix.  

In the case where $\alg$ is noncommutative, this determinant can be defined in several ways.  We show that for estimators based on the conventional determinant, the critical ratio of the second moment to the square of the first---and therefore the number of trials we need to obtain a good estimate of the permanent---is $(1 + O(1/d))^n$ when $\alg$ is the algebra of $d \times d$ matrices. These results can be extended to group algebras, and semi-simple algebras in general.

We also study the \emph{symmetrized} determinant of Barvinok, showing that the resulting estimator has small variance when $d$ is large enough. However, if $d$ is constant---the only case in which an efficient algorithm is known---we show that the critical ratio exceeds $2^{n} / n^{O(d)}$. Thus our results do not provide a new polynomial-time approximation scheme for the permanent.  Indeed, they suggest that the algebraic approach to approximating the permanent faces significant obstacles.

We obtain these results using diagrammatic techniques in which we express matrix products as contractions of tensor products.  When these matrices are random, in either the Haar measure or the Gaussian measure,  we can evaluate the trace of these products in terms of the cycle structure of a suitably random permutation.  In the symmetrized case, our estimates are then derived by a connection with the character theory of the symmetric group.
\end{abstract}

\section{Introduction}
\label{sec:intro}

The \emph{permanent} of an $n \times n$ matrix $A$ is $\perm A = \sum_{\pi \in S_n} \prod_{i=1}^n A_{i,\pi i}$, 
where $S_n$ denotes the group of permutations of $n$ objects.  If $A_{ij} \in \{0,1\}$ for all $i,j$, we can also write
$\perm A = \abs{ \{ \pi \in S_n \mid \pi \vdash A \} }$
where $\pi \vdash A$ denotes the following relation, 
\[
\pi \vdash A \Leftrightarrow A_{i,\pi i}=1 \text{ for all $i$} \, .
\]
Computing the permanent of a $\{0,1\}$ matrix is $\textsl{\#P}$-complete.  Therefore, we cannot expect to compute it efficiently without startling complexity-theoretic consequences, including the collapse of the polynomial hierarchy~\cite{Valiant1979,Toda}. 

On the other hand, Godsil and Gutman~\cite{godsil-gutman} pointed out the following charming fact.  If we define the matrix-valued random variable $M$ so that  $M_{ij} = \rho_{ij} A_{ij}$,  where the $\rho_{ij}$ are chosen independently and uniformly from $\{\pm1\}$, and define $X = (\det M)^2$, then it is easy to check that $X$ is an estimator for $\perm A$, which is to say that $\Exp[ X ] = \perm A$.  Since $\det M$ can be computed efficiently, so can $X$.  This suggests a natural randomized approximation algorithm for the permanent: average a family of independent samples of $X$.

The quality of this approximation can be controlled by determining the variance of $X$.  If $X_t$ denotes the average of $X$ over $t$ independent trials, then Chebyshev's inequality shows that, in order for $X_t$ to yield an approximation of $\Exp[X]$ within a factor $\alpha=O(1)$ with probability $\Omega(1)$, the number of trials we need is at most
\[
t \sim \frac{\Var X}{\Exp[X]^2} \le \frac{\Exp[X^2]}{\Exp[X]^2} \, .
\]
Following~\cite{chien-rasmussen-sinclair}, we refer to this quantity as the \emph{critical ratio} of the estimator. Karmarkar, Karp, Lipton, Lov{\'a}sz, and Luby~\cite{karmarkar-etal} showed, unfortunately, that the critical ratio is $3^{n/2}$ in the worst case, ignoring $\poly(n)$ factors.  Then again, they showed that we can decrease this to $2^{n/2}$ by drawing $\rho_{ij}$ uniformly from the unit circle in the complex plane, or simply from the cube roots of unity, instead of $\{\pm1\}$.  Barvinok~\cite{barvinok99} obtained a more concentrated estimator by drawing $\rho_{ij}$ from normal distributions over $\R$, $\C$, and the quaternions $\mathbb{H}$.

This raises the interesting possibility that, by choosing the $\rho_{ij}$ from the right set of algebraic objects, we might be able to reduce the critical ratio to $\e^{o(n)}$, or even to $\poly(n)$, resulting in a subexponential or polynomial-time algorithm.  One exciting result in this direction is due to Chien, Rasmussen, and Sinclair~\cite{chien-rasmussen-sinclair}, who showed that certain determinants defined over the Clifford algebra $\textsc{CL}_k$ with $k$ generators give estimators where the critical ratio is $(1+O(2^{-k/2}))^{n/2}$.  In the case of the quaternion group, where $k=3$, they gave a polynomial-time algorithm for a type of determinant where the critical ratio grows as $(3/2)^{n/2}$.  This is currently the best known critical ratio for an algebraic estimator which can be computed efficiently.  Sadly, however, for larger $k$ we do not know how to compute these determinants in polynomial time.

These results can be given a uniform presentation by defining a notion of determinant for a matrix $M$ over an associative algebra $\alg$.  The traditional Cayley determinant is then
\begin{equation}
\label{eq:det}
\det M = \sum_{\alpha \in S_n} (-1)^{\alpha} \prod_{i=1}^n M_{i, \alpha i} \, ,
\end{equation}
where $(-1)^\alpha$ denotes the sign of the permutation $\alpha$. Note that $\det M$ takes values in $\alg$. If $\alg$ is noncommutative, however, the determinant as defined in~\eqref{eq:det} may depend on the order in which the product is taken.  As written, each traversal $M_{i, \alpha i}$ is ordered from the top row to the bottom row; we could just as easily order them from the left column to the right.  This introduces some arbitrariness to the definition, and appears to complicate the problem of computing such determinants, even when the algebra $\alg$ has small dimension~\cite{Nisan91}. 

One natural remedy is to remove this order dependence by forcibly symmetrizing each product appearing in~\eqref{eq:det}.  This gives the following \emph{symmetrized determinant},
\begin{equation}
\label{eq:sdet}
\sdet M = \frac{1}{n!} \sum_{\alpha,\alpha' \in S_n} (-1)^{\alpha' \alpha^{-1}} \prod_{i=1}^n M_{\alpha i, \alpha' i} \, .
\end{equation}
Observe that $\sdet$ is obtained by symmetrizing each product appearing in~\eqref{eq:det}.  This definition is due to Barvinok~\cite{barvinok-sdet}, who showed that if $\alg$ has dimension $m$, the symmetrized determinant can be computed in time $O(n^{m+O(1)})$.  In contrast, no efficient algorithm is currently known for the unsymmetrized Cayley determinant~\eqref{eq:det}, even when the dimension of $\alg$ is constant.

We focus on the algebra $\alg_d$, consisting of all $d \times d$ matrices over $\C$.  We remark that any finite dimensional $C^*$-algebra,\footnote{A $C^*$-algebra is an algebra over $\R$ or $\C$ possessing a norm $\| \cdot \|$ and an involution operator $\cdot^*$ consistent in the sense that $\| x^* x\|^2 = \| x\|^2.$ See, e.g., \cite[\S1]{Conway:Operator} for a complete definition.} which appear to be the natural settings for such approximations, are \emph{semi-simple}, meaning that they can be decomposed as a direct product of algebras of the form $\alg_d$. In particular, all group algebras and the Clifford algebras studied in~\cite{chien-rasmussen-sinclair} have this property. It follows that many of our results, especially lower bounds on the critical ratio, carry over easily to estimators based on suitable distributions in semisimple algebras.

Now, given a matrix $A$ with entries in $\{0, 1\}$, define $M_{ij} = \rho_{ij} A_{ij}$, where the $\rho_{ij}$ are independently random $d \times d$ matrices.  (We focus on $\{0,1\}$ matrices, but we can let the $A_{ij}$ be arbitrary nonnegative reals by taking $M_{ij} = \rho_{ij} \sqrt{A_{ij}}$.)  Since the $M_{ij}$ take values in $\alg_d$, then so do $\det M$ and $\sdet M$.  There are several ways to turn these matrix-valued determinants into real-valued estimators for the permanent of a real-valued matrix $A$.  As mentioned above, most of the existing literature has focused on the Frobenius norm of these determinants. 
For technical reasons, we focus first on the absolute value squared of their trace.  This gives us two estimators, 
\[
\Xunsym = \abs{ \tr \det M }^2 \quad \text{and} \quad \Xsym = \abs{ \tr \sdet M }^2 \, . 
\]
Note that these are random $\R$-valued variables depending on the $\rho_{ij}$.  We will then address the Frobenius estimators, 
\[
\Xfrobunsym = \norm{ \det M }^2 \quad \text{and} \quad \Xfrobsym = \norm{ \sdet M }^2 \, . 
\]
As an additional degree of freedom, we can draw $\rho_{ij}$ according to two different distributions on $\alg_d$.  The \emph{Haar measure} is the uniform distribution over unitary matrices.  In the \emph{Gaussian measure}, each entry of $\rho_{ij}$ is drawn independently from the Gaussian distribution on $\C$ with mean $0$ and variance $1/d$: that is, its real and imaginary parts are drawn independently from the Gaussian distribution on $\R$ with mean $0$ and variance $1/(2d)$, $p(x) = \e^{-x^2} / \sqrt{\pi}$.

Our main contribution is given by the following theorems.

\begin{theorem}
\label{thm:main}
For both the Haar and Gaussian measures, in the unsymmetrized case we have
\begin{equation}
\label{eq:unsym}
\frac{\Exp[\Xunsym^2]}{\Exp[\Xunsym]^2} = \left( 1+O\!\left( \frac{1}{d} \right) \right)^n \, . 
\end{equation}
In the symmetrized case,
\begin{equation}
\label{eq:sym-constant}
\frac{\Exp[\Xsym^2]}{\Exp[\Xsym]^2} \leq 2^{2n} \,n^{-d+O(1)} \mbox{ if $d=O(1)$}
\end{equation}
and more generally, 
\begin{equation}
\label{eq:sym-growing}
\frac{\Exp[\Xsym^2]}{\Exp[\Xsym]^2} = O\!\left(\e^{4n^2/d}\right) \, . 
\end{equation}
\end{theorem}

Additionally, we establish lower bounds on the critical ratio $\Exp[\Xsym^2]/\Exp[\Xsym]^2$.
\begin{theorem}
\label{thm:lower}
Let $A$ be the $n \times n$ identity matrix and $d$ a constant.  Then
\begin{equation*}
\frac{\Exp[\Xsym^2]}{\Exp[\Xsym]^2} = \Omega\!\left(\frac{2^n}{n^d}\right) 
\quad \text{and} \quad 
\frac{\Exp[\Xsym^2]}{\Exp[\Xsym]^2} = \left(1 - O\!\left(\frac{1}{d} \right) \right)^n \Omega\!\left(\frac{2^n}{n^d}\right)\,,
\end{equation*}
when the $\rho_{ij}$ are distributed according to the Gaussian or Haar measure respectively.
\end{theorem}

Finally, we show the critical ratio differs by at most $d^4$ for the Frobenius estimators than for those given by the square of the trace:
\begin{theorem}
\label{thm:frob}
\begin{equation}
\label{eq:frob}
\frac{1}{d^4} \,\frac{\Exp[\Xunsym^2]}{\Exp[\Xunsym]^2}
\le \frac{\Exp[\Xfrobunsym^2]}{\Exp[\Xfrobunsym]^2}
\le d^4 \,\frac{\Exp[\Xunsym^2]}{\Exp[\Xunsym]^2}
\, ,
\end{equation}
and similarly for $\Xfrobsym$.
\end{theorem}

These results give a somewhat frustrating outlook.  The critical ratio for the unsymmetrized estimator behaves very well, becoming more and more mildly exponential as $d$ increases, much like the Clifford group estimator of~\cite{chien-rasmussen-sinclair}.  However, we do not know how to compute these estimators efficiently.  On the other hand, we can compute the symmetrized estimator if $d$ is constant~\cite{barvinok-sdet}, but our results show that its critical ratio does not decrease appreciably until $d$ is roughly $n^2$.

Barvinok~\cite{barvinok-sdet} suggested that the estimators $\Xfrobsym$ might become asymptotically concentrated when $d$ is large, but constant.  Specifically, he made the following conjecture (where we have weakened the lower bound, specialized to $\{0,1\}$ matrices, and changed the notation to fit our purposes):
\begin{conjecture}
\label{conj:barvinok1}
If $A$ is an $n \times n$ matrix with entries in $\{0,1\}$, let $M(A)$ be the matrix $M_{ij} = \rho_{ij} A_{ij}$, where each $\rho_{ij}$ is chosen independently from the Gaussian distribution on $\alg_d$.  Define $M(\one)$ similarly, where $M_{ij} = \rho_{ij} \delta_{ij}$.  Then there is a sequence of constants $\gamma_d$, where $\lim_{d \to \infty} \gamma_d=1$, such that for any $\eps > 0$, 
\[
\lim_{n \to \infty} \Pr\left[ (\gamma_d+\eps)^{-n} \perm A \le \frac{\norm{M(A)}^2}{\norm{M(\one)}^2} \le (\gamma_d+\eps)^n \perm A \right] = 1 \, . 
\]
\end{conjecture}

Our results do not address Conjecture~\ref{conj:barvinok1} directly,  However, given Chebyshev's inequality, it is very natural to consider the following stronger conjecture, which would imply Conjecture~\ref{conj:barvinok1}:
\begin{conjecture}
\label{conj:barvinok2}
There is a sequence of constants $\theta_d$, where $\lim_{d \to \infty} \theta=1$, such that for any $n \times n$ matrix $A$, the critical ratio of the estimator $\Xfrobsym=\norm{M(A)}^2$ obeys
\[
\frac{\Exp[\Xfrobsym^2]}{\Exp[\Xfrobsym]^2} \le \theta_d^n \, .
\]
\end{conjecture}
Sadly, Theorems~\ref{thm:lower} and~\ref{thm:frob} imply that Conjecture~\ref{conj:barvinok2} is false.  It is still conceivable that Conjecture~\ref{conj:barvinok1} is true, but it seems that any proof of it would have to bound higher moments of the estimator: the first and second moments alone do not scale in a way that gives concentration.

The remainder of the paper is organized as follows.  In Section~\ref{sec:exp}, we calculate the expectations of these estimators, showing that they are each a constant $a_d$ times the permanent, and computing the constant explicitly using a diagrammatic technique.  In Sections~\ref{sec:second-unsym} and~\ref{sec:second-sym}, we bound their second moments using the same technique, proving Theorems~\ref{thm:main} and~\ref{thm:lower}.  In Section~\ref{sec:frob}, we relate the critical ratio for the Frobenius estimators to the trace-squared estimators, proving Theorem~\ref{thm:frob}.
Finally, in Section~\ref{sec:conclusion} we discuss the implications of this theorem, and the remaining barriers to an algebraic approximation scheme for the permanent.

\section{The expectation}
\label{sec:exp}

Before we proceed, we write the following expansions for these estimators, which we will find useful for calculating their expectations and second moments:
\begin{align}
\Xunsym 
&= \sum_{\alpha,\beta \vdash A} 
(-1)^{\alpha \beta}
\left( \tr \prod_i \rho_{i, \alpha i} \right)
\left( \tr \prod_i \rho^*_{i, \beta i} \right) 
\label{eq:det-expansion} \\
\Xsym 
&= \sum_{\kappa,\lambda \vdash A} 
(-1)^{\kappa \lambda} \,
\Exp_{\alpha,\beta} \!
\left( \tr \prod_i \rho_{\alpha i, \kappa \alpha i} \right)
\left( \tr \prod_i \rho^*_{\beta i, \lambda \beta i} \right) 
\, .
\label{eq:sdet-expansion}
\end{align}

Since $\Exp[\rho_{ij}] = 0$, any term in which some $\rho_{ij}$ appears only once will have zero expectation.  Then the cross-terms in the expansion~\eqref{eq:det-expansion} are zero in expectation except when $\alpha=\beta$, so 
\begin{align}
\Exp[\Xunsym] 
 &= \sum_{\alpha \vdash A} \Exp \!\left(\tr \prod_i \rho_{i,\alpha i}\right)\left(\tr \prod_i \rho_{i,\alpha i}^*\right) 
= \sum_{\alpha \vdash A} \Exp \abs{ \tr \prod_i \rho_{i,\alpha i} }^2 
 = \perm A \, .
\label{eq:exp-unsym}
\end{align}
Here we used the use the following fact, which is an easy exercise: if $\sigma$ is the product of $n$ independent random matrices, chosen from the Haar measure or the Gaussian measure, then $\Exp \abs{ \tr \sigma }^2 = 1$. 

Similarly, the only terms in~\eqref{eq:sdet-expansion} that contribute to $\Exp[\Xsym]$ are those where $\lambda=\kappa$, so that each $\rho_{ij}$ appears twice or not at all.  Thus
\begin{align}
\Exp_{\{\rho_{ij}\}} [\Xsym]
&= \sum_{\kappa \vdash A} 
\Exp_{\{\sigma_i\}}
\Exp_{\alpha,\beta} 
\left( \tr \prod_i \sigma_{\alpha i} \right)
\left( \tr \prod_i \sigma^*_{\beta i} \right) 
= a_d \cdot \perm A 
\nonumber \\
\mbox{where} \quad 
a_d &=
\Exp_{\{\sigma_i\}}
\Exp_{\alpha,\beta} 
\left( \tr \prod_i \sigma_{\alpha i} \right)
\left( \tr \prod_i \sigma^*_{\beta i} \right)\, .
\label{eq:exp-sym}
\end{align}
A similar result for the Frobenius estimator $\Xfrobsym = \norm{\sdet M}^2$ appears as Theorem 4.3 in Barvinok~\cite{barvinok-sdet}.

We can think of $a_d$ as the expectation, over all pairs of permutations $\alpha, \beta$, of the covariance between the trace of two products of the same $n$ random matrices, where the products are taken in the orders given by $\alpha$ and $\beta$.  This expectation clearly stays the same if we assume that $\alpha$ is the identity $1$ and $\beta$ is uniformly random, so we can write 
\[
a_d =
\Exp_{\beta} \Exp_{\{\sigma_i\}} 
\left( \tr \prod_i \sigma_{i} \right)
\left( \tr \prod_i \sigma^*_{\beta i} \right) \, .
\]

We will evaluate these covariances using a diagrammatic approach.  First, suppose we have $n$ linear operators $\sigma_1, \ldots, \sigma_n$.  The trace of their product is 
\[
(\sigma_1)^{i_1}_{i_2} (\sigma_2)^{i_2}_{i_3} \cdots (\sigma_n)^{i_n}_{i_1} \, . 
\]
Here we save ink by using the Einstein summation convention, in which any index that appears twice is automatically summed over.  
We can think of this product as a particular kind of internal trace of the tensor product
\[
(\sigma_1 \otimes \cdots \otimes \sigma_n)^{i_1, \ldots, i_n}_{j_1, \ldots, j_n}
= (\sigma_1)^{i_1}_{j_1} \cdots (\sigma_n)^{i_n}_{j_n} \, , 
\]
where we contract the index $i_t$ with $j_{(i+1) \bmod n}$ for each $i$.  We draw this on the left-hand side of Fig.~\ref{fig:internal}.  Then if $n=3$, say, and $\beta$ is the the transposition $(2\, 3)$, the covariance 
\[
\Exp_{\sigma_1,\sigma_2,\sigma_3} \left( \tr \sigma_1 \sigma_2 \sigma_3 \right) \left( \tr \sigma_1 \sigma_3 \sigma_2 \right)^*
\]
becomes a certain contraction of the tensor product of three independent and identical expectations, 
\begin{equation}
\label{eq:3cupcaps}
 \Exp_{\sigma_1} \left[ \sigma_1 \otimes \sigma_1^* \right] \otimes
 \Exp_{\sigma_2} \left[ \sigma_2 \otimes \sigma_2^* \right] \otimes
 \Exp_{\sigma_3} \left[ \sigma_3 \otimes \sigma_3^* \right] \, . 
\end{equation}

The following lemma is well-known;  we prove it in Appendix~\ref{sec:cupcap} for completeness.

\begin{lemma}
\label{lem:cupcap}
If $\sigma$ is chosen according to the Haar measure or the Gaussian measure, then
\begin{equation}
\label{eq:cupcap-coord}
\Exp_\sigma \left[ \sigma \otimes \sigma^* \right]^{ik}_{j\ell} 
= \frac{1}{d} \,\delta^{ik} \delta_{j\ell} \, . 
\end{equation}
\end{lemma}

\noindent
We can represent~\eqref{eq:cupcap-coord} diagrammatically as a ``cupcap,'' 
\begin{equation}
\label{eq:cupcap}
 \Exp_\sigma \left[ \sigma \otimes \sigma^* \right]
 = \frac{1}{d} \,\cupcap \, . 
\end{equation}
A tensor product such as~\eqref{eq:3cupcaps} becomes three cupcaps side by side, and contracting it consists of connecting pairs of inputs and outputs until the diagram becomes closed.  For instance, the expectation of $\left( \tr \sigma_1 \sigma_2 \sigma_3 \right) \left( \tr \sigma_1 \sigma_3 \sigma_2 \right)^*$ corresponds to the diagram on the right-hand side of Fig.~\ref{fig:internal}. Here we have drawn the cupcaps on the top and bottom of the diagram (between corresponding indices of $\sigma_i$ and $\sigma_i*$) and the connections between them in the interior.  

When we evaluate the trace of this diagram, each of the $n$ cupcaps introduces a factor of $1/d$ according to~\eqref{eq:cupcap}, and each loop in the diagram corresponds to an index which can be set independently to any value between $1$ and $d$.  So, the diagram evaluates to $d^{c-n}$ where $c$ is the number of loops.  In this case $n=3$ and $c=1$, and the covariance is $1/d^2$.

\begin{figure}
\centering
\includegraphics[width=0.7 \textwidth]{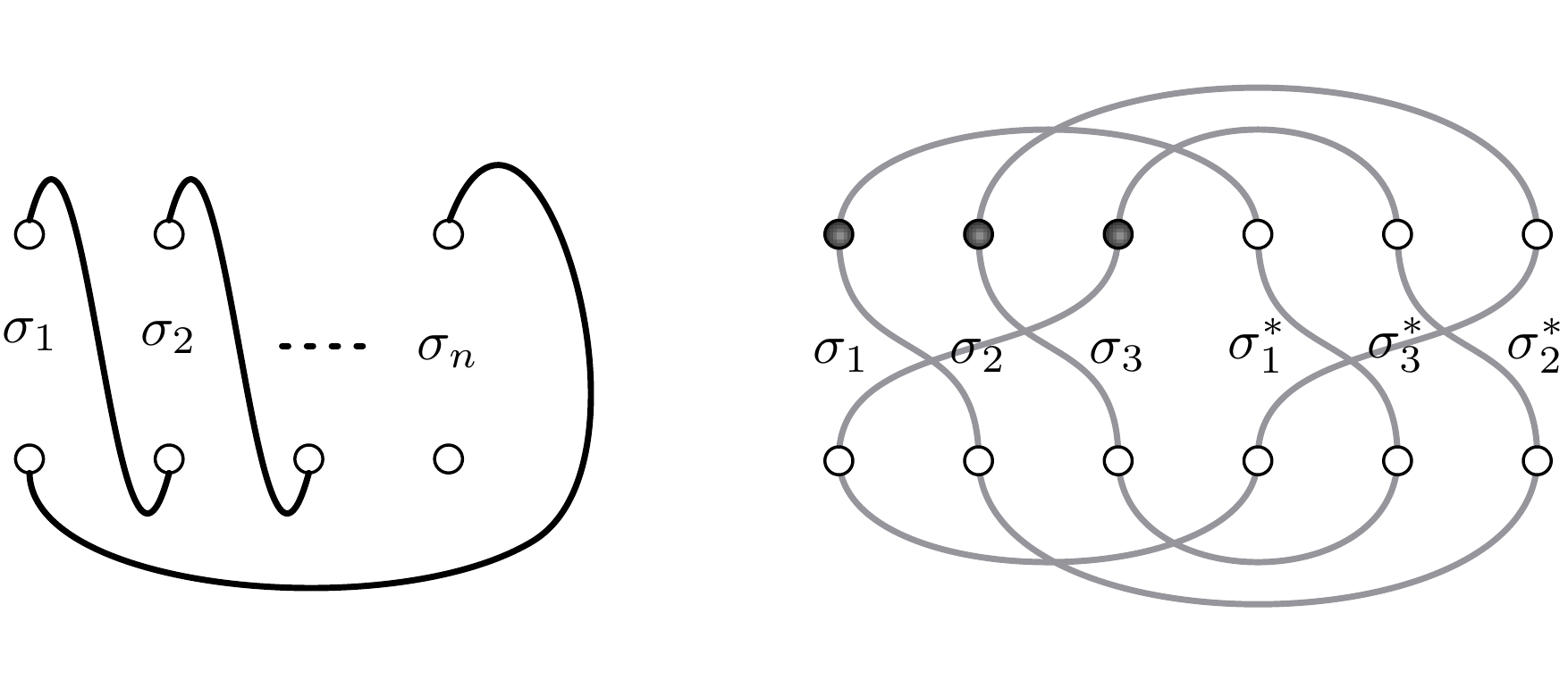}
\caption{The trace of the matrix product $\sigma_1 \sigma_2 \cdots \sigma_n$ is a contraction of the tensor product $\sigma_1 \otimes \cdots \otimes \sigma_n$.  Combining this with~\eqref{eq:cupcap} shows that the covariance between the traces of two permuted products is given by $d^{c-n}$ where $c$ is the number of loops in a diagram like that on the right.  In this case, the covariance between $\tr  \sigma_1 \sigma_2 \sigma_3$ and $\tr \sigma_2 \sigma_1 \sigma_3$ is $1/d^2$, since $n=3$ and $c=1$.}
\label{fig:internal}
\end{figure}

More generally, we can write the covariance between $\tr \prod_i \sigma_i$ and $\tr \prod_i \sigma_{\beta i}$ as a function of $\beta$ as follows.  The cupcaps match the upper indices of the $\sigma$s in the first product to those of the second product according to $\beta$, and the lower indices of the second product to those of the first product according to $\beta^{-1}$.  If $r$ denotes the rotation $(1\, 2\, \cdots\, n)$, which ``weaves'' the $\sigma$s together and takes the trace of their product, then following the diagram around gives a permutation on (say) the upper $n$ indices of the first product (darkened in Fig.~\ref{fig:internal}) equal to the commutator
$[\beta,r] = \beta r \beta^{-1} r^{-1}$. 
Each loop in the diagram corresponds to a cycle in this permutation.  So, we have
\begin{equation}
\label{eq:cd-c}
a_d = \frac{1}{d^n} \Exp_\beta d^{c([\beta,r])} 
\end{equation}
where $c(\pi)$ denotes the number of cycles in a permutation $\pi$.
Note that we always have
\begin{equation}
\label{eq:sd-lower}
a_d \ge \frac{n}{n!} \, . 
\end{equation}
This follows because, with probability $n/n!$, a uniformly random $\beta$ is one of the $n$ powers of $r$.  In that case $[\beta,r]=1$, and $d^{c([\beta,r])}=d^n$.  It can be shown that this bound is tight when $d = \omega(n^2)$.

The expectation~\eqref{eq:cd-c} can be viewed as the inner product of $P_n$, the uniform distribution over the conjugacy class $[r] = \{ \pi^{-1} r \pi \mid \pi \in S_n\}$, and the function $d^{c(\cdot)}$, both of which are \emph{class functions}---invariant under conjugation. Below, we show that these can be expanded in terms of the characters of the group $S_n$ and analyzed using the Littlewood-Richardson rule; this yields an exact expression for $a_d$: 
\begin{lemma}
\label{lem:sd-exact}
\begin{equation}
\label{eq:sd-exact}
\text{If $d \le n$,} \;
a_d = \frac{1}{d^n} \binom{n+d}{n+1} \, .
\quad
\text{If $d > n$,} \;
a_d = \frac{1}{d^n} \left( \binom{n+d}{n+1} - \binom{d}{n+1} \right) \, .
\end{equation}
\end{lemma}

\begin{proof}
First, note that the function $d^{c(\pi)}$ is a \emph{class function}, i.e., one which is invariant under conjugation.  Therefore, in $\Exp_\beta d^{c([\beta,r])}$ we can replace $[\beta,r]$ with $\zeta [\beta,r] \zeta^{-1}$ where $\beta$ and $\zeta$ are uniformly random.  Since
\[
\zeta [\beta,r] \zeta^{-1} 
= \zeta \beta r \beta^{-1} r^{-1} \zeta^{-1}
= \big( (\zeta \beta) r (\zeta \beta)^{-1} \big) \big( \zeta r^{-1} \zeta^{-1} \bigr) \, ,
\]
we can treat this as the expectation of $d^{c(\pi)}$ where $\pi$ is the product of two uniformly random elements of $[r]$, the conjugacy class consisting of cycles of length $n$.  In other words, if $P_n: S_n \rightarrow \R$ is the uniform distribution on the conjugacy class of $n$-cycles, then 
\begin{equation}
\label{eq:sd}
a_d = \frac{1}{d^n} \sum_\pi (P_n * P_n)(\pi) \,d^{c(\pi)}
\end{equation}
where $P_n * P_n$ is the convolution of $P_n$ with itself, 
\[
(P_n * P_n)(\pi) = \sum_{\eta \in S_n} P_n(\eta) \,P_n(\eta^{-1} \pi) \, . 
\]

We will view~\eqref{eq:sd} as an inner product over $S_n$,
\begin{equation}
\label{eq:sd-inner}
a_d 
= \frac{n!}{d^n} \inner{P_n * P_n}{d^{c(\cdot)}} \, ,
\end{equation}
where the inner product over a group $G$ of two functions $f_1,f_2:G \to \C$ is defined as
\[
\inner{f_1}{f_2} = \frac{1}{|G|} \sum_{g \in G} f_1(g)^* f_2(g) \, . 
\]
To evaluate~\eqref{eq:sd-inner}, we will expand $P$ and $d^{c(\cdot)}$ in the Fourier basis, as a sum of irreducible characters of $S_n$. Recall that the characters of a finite group are orthonormal under the inner product above and, additionally, convolution is transformed to pointwise product in the Fourier basis.  In short, for two characters $\chi$ and $\psi$, 
\begin{equation}
\label{eq:characters}
\chi * \psi = \begin{cases}
\frac{|G|}{\chi(1)}\; \chi &\text{if}\;\chi = \psi,\\
0 & \text{if}\;\chi \neq \psi,
\end{cases}
\quad \text{and}\quad
\inner{\chi}{\psi} = \begin{cases}1 & \text{if}\;\chi= \psi, \\ 0 & \text{if}\;\chi \neq \psi.
\end{cases}
\end{equation}

Each character of the symmetric group is associated with a Young diagram, i.e., a partition $\lambda_1 \ge \lambda_2 \ge \cdots$ where $\sum_i \lambda_i = n$.  
In light of the Murnaghan-Nakayama rule (Lemma~\ref{lem:mn} of Appendix~\ref{sec:representation-theory}), the uniform distribution $P_n$ over the conjugacy class $[r]$ is supported solely on \emph{hooks}, i.e., Young diagrams consisting of a single ribbon of size $n$.  Let  $\Lambda_t$ denote the hook of height $t+1$, in which $\lambda_1 = n-t$ for some $0 \le t < n$ and $\lambda_i = 1$ for $1 < i \le t+1$. Let $\chi_t$ denote the corresponding character.  Then $\dim \chi_t = \chi_t(1) = {n-1 \choose t}$ and, again appealing to Lemma~\ref{lem:mn}, we have $\chi_t([r]) = (-1)^t$.  Applying~\eqref{eq:characters} then gives
\begin{equation}
\label{eq:chit-p}
\inner{P_n}{\chi_t} = \frac{(-1)^t}{n!} 
\quad \text{and} \quad
\inner{P_n * P_n}{\chi_t} = \frac{1}{n!} \frac{1}{\binom{n-1}{t}} \, . 
\end{equation}

To calculate the inner product $\inner{d^{c(\cdot)}}{\chi_t}$, consider the following combinatorial representation of $S_n$.  Let $\Sigma$ be the set of strings of length $n$ over the alphabet $\{1,\ldots,d\}$, and let $S_n$ act on $\Sigma$ in the natural way, by permuting the symbols in a given string.  Given a permutation $\pi$, the character $\chi_\Sigma(\pi)$ is the number of strings fixed by $\pi$.  Since each of $\pi$'s cycles can be given an independent label in $\{1,\ldots,d\}$, we have $\chi_\Sigma(\pi) = d^{c(\pi)}$.  

It follows that $\inner{d^{c(\cdot)}}{\chi_t} = \inner{\chi_t}{\chi_\Sigma}$, the number of copies of $\Lambda_t$ appearing in the decomposition of $\Sigma$ into irreducible representations.  To find this, we first decompose $\Sigma$ into a direct sum of combinatorial representations $\Sigma_{(n_1,\ldots,n_d)}$, consisting of strings where $i$ appears $n_i$ times for each $i \in \{1,\ldots,d\}$.  Then $\inner{\chi_{(n_1,\ldots,n_d)}}{\chi_t}$ is given by a \emph{Kostka number}, defined as follows.  First, sort the $n_i$ in decreasing order so that they form a Young diagram $N$.  Then $K^{\Lambda_t}_N = \inner{\chi_{(n_1,\ldots,n_d)}}{\chi_t}$ is the number of semistandard tableaux of shape $\Lambda_t$ and content $N$: that is, the number of ways to fill $\Lambda_t$ with $n_i$ $i$s for each $i \in \{1,\ldots,d\}$, where each row is nondecreasing and where each column is strictly increasing.  

Since $\Lambda_t$ is a hook, to specify a semistandard tableau with a given content it suffices to specify the content of the leftmost column.  Since this column must be strictly increasing, its $t+1$ entries must be distinct.  If $N$ has $k$ rows, i.e., if $n_i \ne 0$ for $k$ values of $i$, then the first one must appear in the top cell, but the remaining $t$ cells can be chosen arbitrarily.  Thus $K^{\Lambda_t}_N = {k-1 \choose t}$, and is $0$ if $t \ge k$.  There are ${d \choose k} {n-1 \choose k-1}$ partitions $(n_1,\ldots,n_d)$ with $k$ nonzero $n_i$.  Since $1 \le k \le \min(d,n)$, summing over $k$ then gives
\begin{equation}
\label{eq:chit-dcpi}
\inner{d^{c(\pi)}}{\chi_t}
= \sum_k {d \choose k} {n-1 \choose k-1} K^{\Lambda_t}_N 
= \sum_{k=1}^{\min(d,n)} {d \choose k} {n-1 \choose k-1} {k-1 \choose t} \, .
\end{equation}

We can now calculate the inner product $\inner{P_n}{d^{c(\cdot)}}$.  Combining~\eqref{eq:chit-p} and~\eqref{eq:chit-dcpi} and summing over $t$, we have
\begin{align*}
n! \inner{P_n}{d^{c(\cdot)}} 
&= n! \sum_{t=0}^{n-1} \inner{P * P}{\chi_t} \inner{d^{c(\cdot)}}{\chi_t} = \sum_{t=0}^{n-1} \sum_{k=1}^{\min(d,n)} 
{d \choose k} {n-1 \choose k-1} {k-1 \choose t} \Big\slash {n-1 \choose t} \\
&= \sum_{k=1}^{\min(d,n)}  {d \choose k} \sum_{t=0}^{k-1} {n-t-1 \choose n-k} = \sum_{k=1}^{\min(d,n)}  {d \choose k} {n \choose k-1} = {n+d \choose n+1} - {d \choose n+1} \, ,
\end{align*}
where ${d \choose n+1}=0$ if $d \le n$.  Combining this with~\eqref{eq:sd-inner} completes the proof.
\end{proof}

\section{The second moment in the unsymmetrized case}
\label{sec:second-unsym}

Squaring~\eqref{eq:det-expansion}---and, for aesthetic reasons, placing the conjugated $\rho$s in the second half of the expression and changing the names of the permutations---gives
\begin{equation}
\label{eq:xunsym2}
\Xunsym^2 = \sum_{\kappa,\lambda,\mu,\nu \vdash A} 
(-1)^{\kappa \lambda \mu \nu} 
\left( \tr \prod_i \rho_{i, \kappa i} \right)
\left( \tr \prod_i \rho_{i, \lambda i} \right) 
\left( \tr \prod_i \rho^*_{i, \mu i} \right)
\left( \tr \prod_i \rho^*_{i, \nu i} \right) 
\, .
\end{equation}
Now we take the expectation over the $\rho_{ij}$.  As before, the only terms of this sum that contribute to this expectation are those in which each $\rho_{ij}$ appears an even number of times.  Moreover, each $\rho_{ij}$ must appear an equal number of times conjugated (in the first and second products) and unconjugated (in the third and fourth products), since $\Exp_\sigma[\sigma \otimes \sigma]=0$.  In the Gaussian measure, this is because $\Exp[(\sigma^i_j)^2]=0$ if $\sigma^i_j$ is chosen from the Gaussian distribution on $\C$.  In the Haar measure, the same thing is true because the tensor square $\sigma \otimes \sigma$ of the defining representation of $\U(d)$ contains no copies of the trivial representation.


For each term of~\eqref{eq:xunsym2}, associated with a tuple $(\kappa, \lambda, \mu, \nu)$, we express the total number of occurrences of each $\rho_{ij}$ with an $n \times n$ matrix $C_{ij}$.  In light of the discussion above, for the terms that contribute to the second moment we have $C_{ij}=0$ if $A_{ij}=0$, $C_{ij} \in \{2,4\}$ if $A_{ij}=1$, and $\sum_i C_{ij} = \sum_j C_{ij} = 2n$.  We will denote these conditions as $C \vdash A$.  As in~\cite{karmarkar-etal,chien-rasmussen-sinclair}, we think of $C$ as a ``double cycle cover'' of the bipartite graph described by $A$.  This graph has $n$ vertices on either side, and an edge between the $i$th vertex on the left and the $j$th vertex on the right if and only if $A_{ij}=1$.  Each vertex has degree $2$ or $4$ in $C$.  Thus $C$ consists of cycles where each edge is covered twice, and possibly some isolated edges which are covered four times.  

We then write the second moment as a sum, over all $C$, of the quadruples such that $(\kappa, \lambda, \mu, \nu) \vdash C$, where this denotes the following relation:
\begin{align*}
(\kappa, \lambda, \mu, \nu) \vdash C \Leftrightarrow \;
&\text{$\pi \vdash A$ for all $\pi \in \{\kappa, \lambda, \mu, \nu\}$} \\
&\text{and} \; \abs{\{ \pi \in \{\kappa, \lambda\} \mid j = \pi i \}} 
 = \abs{\{ \pi \in \{\mu, \nu\} \mid j = \pi i \}}  \; \text{for all} \; i,j \in \{1,\ldots,n\}\, ,\\
&\text{and} \; \abs{\{ \pi \in \{\kappa, \lambda, \mu, \nu\} \mid j = \pi i \}} = C_{ij} \; \text{for all} \; i,j \in \{1,\ldots,n\} \, . 
\end{align*}
In our discussion below, we will treat each $(\kappa, \lambda, \mu, \nu)$ as a ``coloring'' of $C$.  Each double edge is colored $(\kappa,\mu)$, $(\kappa,\nu)$, $(\lambda,\mu)$, or $(\lambda,\nu)$, indicating some pair $\rho_{ij}, \rho^*_{ij}$ appearing in the first and third products, or the first and fourth, and so on.  Each cycle in $C$ must alternate between $(\kappa,\mu)$ and $(\lambda,\nu)$ or between $(\kappa,\nu)$ and $(\lambda,\mu)$.  The isolated edges in $C$ bear all four colors, indicating that some $\rho_{ij}$ appears in all four products.  
We observe that for those tuples that contribute to the second moment, the parity $(-1)^{\kappa \lambda \mu \nu}$ is always 1.

\begin{lemma}
\label{lem:parity}
If $(\kappa, \lambda, \mu, \nu) \vdash C$ for some $C$, then $(-1)^{\kappa \lambda \mu \nu} = 1$.
\end{lemma}

\begin{proof}
Observe that $(-1)^{\kappa \lambda \mu \nu} = (-1)^\pi$ where $\pi = \kappa^{-1} \mu \lambda^{-1} \nu$.  
We claim that the constraints we describe above imply that $\pi=1$.  Consider a cycle $c$ of $C$ on the bipartite graph defined by $A$.  We can view $\kappa, \lambda, \mu, \nu$ as one-to-one mappings from the $n$ vertices on the left side to the $n$ vertices on the right.  If $c$ alternates between $(\kappa, \mu)$ and $(\lambda, \nu)$, then restricting to the vertices on the left side of $c$ we have $\kappa = \mu$ and $\lambda=\nu$.  Similarly, if $c$ alternates between $(\kappa, \nu)$ and $(\lambda, \mu)$, then restricting to these vertices gives $\kappa = \nu$ and $\lambda = \mu$.  Finally, for an isolated edge we have $\kappa = \lambda = \mu = \nu$ when restricted to its left endpoint.  In all cases we have $\kappa^{-1} \mu \lambda^{-1} \nu=1$.
\end{proof}

\noindent
Thus the second moment of the unsymmetrized estimator can be written
\begin{equation}
\label{eq:second-moment-unsym}
\Exp[\Xunsym^2] = \sum_{C \vdash A} \sum_{(\kappa,\lambda,\mu,\nu) \vdash C} 
\Exp_{\{\rho_{ij}\}}
\left( \tr \prod_i \rho_{i, \kappa i} \right)
\left( \tr \prod_i \rho_{i, \lambda i} \right) 
\left( \tr \prod_i \rho^*_{i, \mu i} \right)
\left( \tr \prod_i \rho^*_{i, \nu i} \right) 
\, .
\end{equation}

Many terms in this expectation can be evaluated using the same picture we gave for the expectation.  Each pair $\rho_{ij}, \rho^*_{ij}$ creates a cupcap matching a pair of indices in one of the first two products with a pair in one of the second two products.  However, the isolated edges in $C$ correspond to a fourth-order operator $\Exp_\sigma (\sigma \otimes \sigma \otimes \sigma^* \otimes \sigma^*)$ which we calculate in the following lemma.

\begin{lemma}
\label{lem:fourth}
If $\sigma$ is chosen according to the Gaussian measure, then
\begin{equation}
\label{eq:cupcap-coord-gaussian}
\Exp_\sigma \left[ \sigma \otimes \sigma \otimes \sigma^* \otimes \sigma^* \right]^{ikmp}_{j\ell nq} 
= \frac{1}{d^2} \left( 
\delta^{im} \delta_{jn} \delta^{kp} \delta_{\ell q} 
+ \delta^{ip} \delta_{jq} \delta^{km} \delta_{\ell n} 
\right) \, ,
\end{equation}
or diagrammatically, 
\begin{equation}
\label{eq:cupcap-diag-gaussian}
\Exp_\sigma \left[ \sigma \otimes \sigma \otimes \sigma^* \otimes \sigma^* \right]
= \frac{1}{d^2} \left( \, \cupcapcross + \cupcapnest \, \right)  \, .
\end{equation}
If $\sigma$ is chosen according to the Haar measure, then
\begin{equation}
\label{eq:cupcap-diag-haar}
\frac{1-O(1/d)}{d^2} \left( \, \cupcapcross + \cupcapnest \, \right) \; \preceq \;
\Exp_\sigma \left[ \sigma \otimes \sigma \otimes \sigma^* \otimes \sigma^* \right]
\; \preceq \; 
\frac{1+O(1/d)}{d^2} \left( \, \cupcapcross + \cupcapnest \, \right)  \, ,
\end{equation}
where we write $A \preceq B$ if $B-A$ is positive semidefinite.
\end{lemma}

\begin{proof}
We have
\[
\Exp_\sigma \left[ \sigma \otimes \sigma \otimes \sigma^* \otimes \sigma^* \right]^{ikmp}_{j\ell nq} 
= \Exp\left[ \sigma^i_j \sigma^k_\ell (\sigma^m_n)^* (\sigma^p_q)^* \right]\, .
\]
In the Gaussian measure, if $i=m$, $j=n$, $k=p$, and $\ell=q$, but $i \ne k$ or $j \ne \ell$, this gives $\bigl| \sigma^i_j\bigr|^2 \bigl|\sigma^k_\ell\bigr|^2 = 1/d^2$.  If $i=p$, $j=q$, $k=m$, and $\ell=n$, but $i \ne k$ or $j \ne \ell$, we get the same result.  Finally, if $i=k=m=p$ and $j=\ell=n=p$, we get $\Exp \bigl| \sigma^i_j\bigr|^4 = 2/d^2$.

In the Haar measure, analogous to Lemma~\ref{lem:cupcap} we will calculate the expectation of $\sigma \otimes \sigma \otimes \sigma^* \otimes \sigma^*$ by considering tensor powers of the defining representation $\sigma$ of $\U(d)$.  The tensor square $\sigma \otimes \sigma$ decomposes into symmetric and antisymmetric subspaces, each of which is irreducible:
\[
\sigma \otimes \sigma = \tau_\sym \oplus \tau_\asym \, .
\]
The dimension of $\tau_\pm$ is $d_\pm = (d^2 \pm d)/2$.  We can write the projection operators onto $\tau_\pm$ in terms of the exchange operator $\exchange$ which reverses the order of the tensor product, and the identity $\identity$ \,:
\[
\Pi_\pm = \frac{1}{2} \left( \, \identity \pm \exchange \, \right) \, .
\]
Now writing $\sigma \otimes \sigma \otimes \sigma^* \otimes \sigma^*
= 
( \tau_\sym \oplus \tau_\asym ) \otimes 
( \tau^*_\sym \oplus \tau^*_\asym )$,
the expectation over $\sigma$ is the projection operator onto the trivial subspaces of $\tau_\sym \otimes \tau^*_\sym$ and $\tau_\asym \otimes \tau^*_\asym$:
\begin{equation}
\label{eq:fourth1}
\Exp_\sigma \left[ \sigma \otimes \sigma \otimes \sigma^* \otimes \sigma^* \right]
= \Pi^{\tau_\sym \otimes \tau^*_\sym}_\one \oplus
\Pi^{\tau_\asym \otimes \tau^*_\asym}_\one \, .
\end{equation}
Analogous to~\eqref{eq:cupcap}, we have the handsome 
\begin{equation}
\label{eq:isotypic-1}
\Pi^{\tau_\sym \otimes \tau^*_\sym}_\one = 
(\Pi_\sym \otimes \Pi_\sym)
\cdot \left( \frac{1}{d_\sym} \cupcapcross \, \right)
\cdot (\Pi_\sym \otimes \Pi_\sym)
\end{equation}
and similarly for $\Pi^{\tau_\asym \otimes \tau^*_\asym}_\one$.  Putting these diagrams together with~\eqref{eq:fourth1} gives
\begin{align}
\Exp_\sigma \left[ \sigma \otimes \sigma \otimes \sigma^* \otimes \sigma^* \right]
&= (\Pi_\sym \otimes \Pi_\sym)
\cdot \left( \frac{1}{d_\sym} \cupcapcross \, \right)
\cdot (\Pi_\sym \otimes \Pi_\sym)
+ (\Pi_\asym \otimes \Pi_\asym)
\cdot \left( \frac{1}{d_\asym} \cupcapcross \, \right)
\cdot (\Pi_\asym \otimes \Pi_\asym) 
\nonumber \\
&= \frac{1}{4d_\sym} \left( \, \cupcapcross + \cupcapmixbottom + \cupcapmixtop + \cupcapnest \, \right) 
+ \frac{1}{4d_\asym} \left( \, \cupcapcross - \cupcapmixbottom - \cupcapmixtop + \cupcapnest \, \right) 
\nonumber \\
&= \frac{1}{d^2-1} 
\left( \cupcapcross + \cupcapnest - \frac{1}{d} \left( \, \cupcapmixtop + \cupcapmixbottom \, \right) \right)
\, . 
\label{eq:fourth2}
\end{align}

One can check that~\eqref{eq:fourth2} is the projection operator onto the two-dimensional subspace spanned by the images of 
\[
\cupcapcross \quad \text{and}\quad \cupcapnest \, ,
\]
that is, the vectors $\vec{u} = \frac{1}{d} \sum_{i,j} (i,j,i,j)$ and $\vec{v}=\frac{1}{d} \sum_{i,j} (i,j,j,i)$.  
In general, given two real-valued vectors $\vec{u}$ and $\vec{v}$ of norm $1$, let $\Pi_{\vec{u}}$ and $\Pi_{\vec{v}}$ denote the projection operators onto the subspaces parallel to them, and let $\Pi_{\vec{u},\vec{v}}$ be the projection operator onto the two-dimensional subspace they span.  Then 
\[
\frac{1}{1 + |\langle \vec{u}, \vec{v}\rangle|} (\Pi_{\vec{u}} + \Pi_{\vec{v}}) \; \preceq \; \Pi_{\vec{u},\vec{v}} \;\preceq\; \frac{1}{1-\abs{\inner{\vec{u}}{\vec{v}}}} \left( \Pi_{\vec{u}} + \Pi_{\vec{v}} \right) \, ,
\]
where we write $A \preceq B$ if $B-A$ is positive semidefinite.  To see this, note that the eigenvectors of $\Pi_{\vec{u}} + \Pi_{\vec{v}}$ are $\vec{u} \pm \vec{v}$, with eigenvalues $\lambda_\pm = 1 \pm \inner{\vec{u}}{\vec{v}}$, while their eigenvalues with respect to $\Pi_{\vec{u},\vec{v}}$ are $1$.  
In this case, we have $\inner{\vec{u}}{\vec{v}} = 1/d$ and 
\[
\Pi_{\vec{u}} = \frac{1}{d^2} \cupcapcross \quad \text{and} \quad \Pi_{\vec{v}} = \frac{1}{d^2} \cupcapnest \, .
\]
Thus~\eqref{eq:fourth2} becomes
\begin{equation*}
\left( \frac{1}{1+1/d} \right) \frac{1}{d^2} \left( \, \cupcapcross + \cupcapnest \, \right)
\preceq 
\Exp_\sigma \left[ \sigma \otimes \sigma \otimes \sigma^* \otimes \sigma^* \right]
 \preceq
\left( \frac{1}{1-1/d} \right) \frac{1}{d^2} \left( \, \cupcapcross + \cupcapnest \, \right)\, ,
\end{equation*}
completing the proof.
\end{proof}

The operator $\cupcapcross$ corresponds to the coloring $(\kappa,\mu), (\lambda, \nu)$, in which some $\rho_{ij}$ appears in the first and third products, and another $\rho'_{ij}$ appears in the second and fourth.  Similarly, the operator $\cupcapnest$ corresponds to the coloring $(\kappa,\nu), (\lambda,\mu)$, in which $\rho_{ij}$ appears in the first and fourth products and $\rho'_{ij}$ appears in the second and third.  
Thus Lemma~\ref{lem:fourth} tells us that, with a multiplicative cost of $1+O(1/d)$ per isolated edge in the Haar measure, we can replace a given isolated edge in $C$ with an (unordered) pair of edges.  This pair can be colored in two ways: with $(\kappa,\mu)$ and $(\lambda,\nu)$, or with $(\kappa,\nu)$ and $(\lambda,\mu)$.  Equivalently, we can ``decouple'' each quadruple product $\rho \otimes \rho \otimes \rho^* \otimes \rho^*$ into the sum of two combinations of tensor products, 
\begin{equation}
\label{eq:decouple}
\rho \otimes \rho \otimes \rho^* \otimes \rho^*
\approx 
\rho' \otimes \rho'' \otimes {\rho'}^* \otimes {\rho''}^*
+ \rho' \otimes \rho'' \otimes {\rho''}^* \otimes {\rho'}^* \, , 
\end{equation}
where $\rho'$ and $\rho''$ are chosen independently.

Next we explore the set of $(\kappa,\lambda,\mu,\nu)$ corresponding to a given $C$, or equivalently the set of colorings of $C$.  
We will call a coloring \emph{pure} if every edge in $C$ are colored $(\kappa,\mu)$ or $(\lambda,\nu)$.  This corresponds to pairing the $\rho_{ij}$s in the first product in~\eqref{eq:second-moment-unsym} with their conjugates in the third, and those in the second product with their conjugates in the fourth---and choosing the first term in~\eqref{eq:decouple} for each $\rho_{ij}$ which appears in all four products.  Each cycle in $C$ has two pure colorings, and each isolated edge has one.  Thus the number of pure colorings of $C$ is $2^{t(C)}$ where $t(C)$ is the number of cycles in $C$.  A well-known bijection shows that $(\perm A)^2$ can be written as a sum over cycle covers of the bipartite graph defined by $A$,
\begin{equation}
\label{eq:perm-cycle-sum}
(\perm A)^2 = \sum_{C \vdash A} 2^{t(C)} \, , 
\end{equation}
or equivalently that $(\perm A)^2$ is the total number of pure colorings.  Combining this with~\eqref{eq:exp-unsym}, we have
\begin{equation}
\label{eq:unsym-expsquared}
\Exp[\Xunsym]^2 = \sum_{C \vdash A} \sum_{\substack{(\kappa,\lambda,\mu,\nu) \vdash C \\ {\rm pure}}} 1 \, .
\end{equation}

On the other hand, we can associate each coloring with a pure one, say by replacing the color $(\kappa,\nu)$ with $(\kappa,\mu)$ and $(\lambda,\mu)$ with $(\lambda,\nu)$ on each edge.  If this converts a tuple of permutations $(\kappa,\lambda,\mu',\nu')$ to a tuple $(\kappa,\lambda,\mu,\nu)$ corresponding to a pure coloring, we will write $(\mu',\nu') \vdash (\kappa,\lambda,\mu,\nu)$.  Then, at the risk of some notational overload, we write~\eqref{eq:second-moment-unsym} as a sum over pure colorings:
\begin{equation*}
\Exp[\Xunsym^2] = \sum_{C \vdash A} 
\sum_{\substack{\; (\kappa,\lambda,\mu,\nu) \vdash C \; \\ {\rm pure}}}
\sum_{(\mu',\nu') \vdash (\kappa,\lambda,\mu,\nu)} 
\Exp_{\{\rho_{ij}\}}
\left( \tr \prod_i \rho_{i, \kappa i} \right)
\left( \tr \prod_i \rho_{i, \lambda i} \right) 
\left( \tr \prod_i \rho^*_{i, \mu i} \right)
\left( \tr \prod_i \rho^*_{i, \nu i} \right).
\end{equation*}

Now, analogous to~\cite{karmarkar-etal}, we bound the critical ratio $\Exp[\Xunsym^2]/\Exp[\Xunsym]^2$ as the maximum ratio between corresponding terms in these two sums, associated with some pure coloring of some cycle cover.  The worst possible case is when $C$ consists entirely of isolated edges, since in that case we can switch the colors on each edge independently, giving $2^n$ colorings for the single pure one.

\begin{figure}
\centering
\includegraphics[width=\columnwidth]{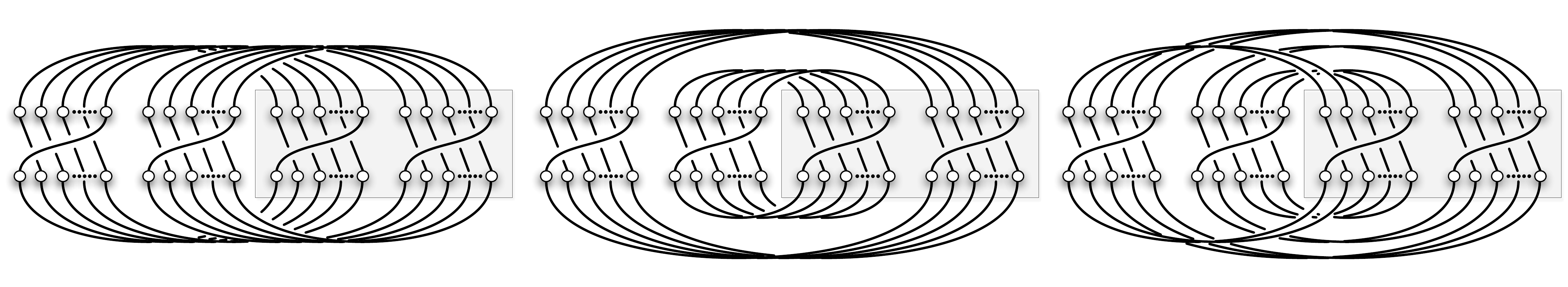}
\caption{Terms corresponding to a given cycle cover $C$, where the $\rho$s in the gray box are conjugated and $n=5$.  Left, a pure coloring, which has $2n$ loops.  Middle, a maximally impure coloring, which also has $2n$ loops.  Right, a mixed coloring corresponding to the string $s=00111$. 
A careful inspection shows that it has $8$ loops: $6$ of length $4$, and $2$ of length $8$.}
\label{fig:unsymmetrized-pure}
\end{figure}

We can parametrize these $2^n$ colorings by strings $s \in \{0,1\}^n$, where $s_i=0$ if the coloring of the $i$th edge is pure, and $1$ if its colors are switched.  This produces diagrams such as those shown in Fig.~\ref{fig:unsymmetrized-pure}, weaving a total of $8n$ vertices together.  As in our calculation of the expectation, the corresponding product of traces is $d^{c-2n}$ where $c$ is the number of loops in this diagram.  

Both the pure and ``completely impure'' colorings $0^n$ and $1^n$---where the $\rho$s in the first product are all paired with those in the third or fourth respectively, and the those in the second product are all paired with those in the fourth or third---have $2n$ loops.  In general, the number of loops is $2n$ minus the number of times $s$ switches back and forth between $0$ and $1$ when $s$ is arranged cyclically.  Specifically, there are two loops of length $4$ for each $i$ where $s_i=s_{(i+1)\bmod n}$, and a cycle of length $8$ for each $i$ where $s_i \ne s_{(i+1)\bmod n}$.  

For each even $i$ with $0 \le i \le n$, there are $2 \binom{n}{i}$ strings which switch back and forth  $i$ times.  Therefore, combined with Lemma~\ref{lem:fourth}, we have (for a cycle  cover $C$ consisting of $n$ isolated edges)
\begin{gather*}
\sum_{(\mu',\nu') \vdash (\kappa,\lambda,\mu,\nu)} 
\Exp_{\{\rho_{ij}\}}
\left( \tr \prod_i \rho_{i, \kappa i} \right)
\left( \tr \prod_i \rho_{i, \lambda i} \right) 
\left( \tr \prod_i \rho^*_{i, \mu i} \right)
\left( \tr \prod_i \rho^*_{i, \nu i} \right)  \\
= \left( 1+O\!\left( \frac{1}{d} \right) \right)^n \times 
2 \!\!\sum_{i=0,2,4,\ldots}^n {n \choose i} d^{-i} 
= \left( 1+O\!\left( \frac{1}{d} \right) \right)^n \times \left( \left( 1+\frac{1}{d} \right)^n + \left( 1-\frac{1}{d} \right)^n \right) 
\end{gather*}
In the Gaussian measure, this expression is exact if we remove the prefactor $(1+O(1/d))^n$; but in any case, we get a bound $(1+O(1/d))^n$ in either measure.  Combining this with~\eqref{eq:unsym-expsquared} 
completes the first part of the proof of Theorem~\ref{thm:main}.

\section{The second moment in the symmetrized case}
\label{sec:second-sym}

Our analysis of the second moment in the symmetrized case proceeds in two steps. We begin, as with the unsymmetrized case, by diagrammatically analyzing the relevant traces. The result is a sum over double cycle covers weighted by an exponential generating function $\sum_{\pi} d^{c(\pi)}$ over a subset of the symmetric group $S_{2n}$. We then show that an allied quantity can be analyzed, as in Lemma~\ref{lem:sd-exact}, by harmonic analysis on $S_{2n}$.

Before stating the main lemmas of this section, we introduce some further notation.  As in~\eqref{eq:perm-cycle-sum}, $t(C)$ denotes the number of cycles in $C$.  
As before, we let $r$ denote the rotation $(1, 2, \ldots, n) \in S_n$. The expression $\pi^\sigma = \sigma^{-1} \pi \sigma$ denotes conjugation, and, for two elements $\pi, \sigma \in S_n$, we let $(\pi, \sigma)$ denote the element of $S_{2n}$ given by applying $\pi$ and $\sigma$ to the first $n$ and last $n$ elements of $\{ 1, \ldots, 2n\}$, respectively. Finally, we let $w_k$ denote the involution $(1\; n+1)(2\; n+2) \cdots (k\; n+k)$ with the convention that $w_0$ is the identity.  We can then write $\Exp[\Xsym^2]$ in terms of the following quantity:
\[
a^{(2)}_d = \sum_{k=0}^n \binom{n}{k} \frac{1}{d^{2n}} 
\Exp_{\alpha, \beta, \gamma, \delta} d^{c\left( (r^{-1},r^{-1})^{(\alpha, \beta)} w_k (r,r)^{(\gamma, \delta)} w_k \right) } \, .
\]

\begin{lemma}
\label{lem:symmetrized-second-diagram}
If the $\rho_{ij}$ are drawn according to the Gaussian or Haar measure, 
\[
\frac{\Exp[\Xsym^2]}{\Exp[\Xsym]^2} \leq \left(1 + O\!\left( \frac{1}{d} \right) \right)^n \frac{a^{(2)}_d}{a_d^2} \, .
\]
\end{lemma}

We delay the proof of Lemma~\ref{lem:symmetrized-second-diagram} just long enough for some comforting words regarding the major remaining obstacle: estimating $a^{(2)}_d$. 
While we do not have a simple, exact expression for $a^{(2)}_d$, we \emph{can} control a larger quantity, 
$$
\tilde{a}^{(2)}_d = \sum_{k=0}^n \binom{n}{k}^2 \frac{1}{d^{2n}} 
\Exp_{\alpha, \beta, \gamma, \delta} d^{c\left( (r^{-1},r^{-1})^{(\alpha, \beta)} w_k (r,r)^{(\gamma, \delta)} w_k \right)} \, ,
$$
in which the $k$th term of the sum is graced with an extra factor of $\binom{n}{k}$. With this reweighting we can analyze $\tilde{a}_d^{(2)}$ in terms of the Fourier expansions of the class function $d^{c(\cdot)}$, determined by the Kostka numbers, and the convolution square of the conjugacy class $\{(r,r)^\sigma \mid \sigma \in S_{2n}\}$, determined by the Murnaghan-Nakayama rule. This results in the following bound.

\begin{lemma}
\label{lem:a-2-hat}
With notation as above,
\begin{equation*}
\frac{1}{\binom{n}{n/2}} \cdot \tilde{a}_d^{(2)} \; \leq\; a^{(2)}_d \;\leq\; \tilde{a}_d^{(2)}
\end{equation*}
and
$$
\frac{1}{d^{2n}} \binom{2n}{n} \binom{2n+d-1}{2n} \:\le\: \tilde{a}^{(2)}_d\: \le\: \frac{4n^2}{d^{2n}} \binom{2n}{n} \binom{2n + d - 1}{2n}\, .
$$ 
\end{lemma}

\noindent
Combining this with Lemmas~\ref{lem:symmetrized-second-diagram} and~\ref{lem:sd-exact} completes the proof of~\eqref{eq:sym-constant} and~\eqref{eq:sym-growing} in Theorem~\ref{thm:main}.

We return now to the proofs of these two lemmas.
\begin{proof}[Proof of Lemma~\ref{lem:symmetrized-second-diagram}]
Squaring~\eqref{eq:sdet-expansion}, the second moment of the symmetrized estimator can be written
\begin{equation}\label{eq:second-moment}
\Exp[\Xsym^2] = \sum_C \sum_{(\kappa,\lambda,\mu,\nu) \vdash C} 
\, \Exp_{\alpha,\beta,\gamma,\delta} 
\Exp_{\{\rho_{ij}\}} \!
\left( \tr \prod_i \rho_{\alpha i, \kappa \alpha i} \right)
\left( \tr \prod_i \rho_{\beta i, \lambda \beta i} \right) 
\left( \tr \prod_i \rho^*_{\gamma i, \mu \gamma i} \right)
\left( \tr \prod_i \rho^*_{\delta i, \nu \delta i} \right) .
\end{equation}
Consider now a term of~\eqref{eq:second-moment} corresponding to a tuple $(\kappa, \lambda, \mu, \nu)$ of the form
\begin{equation}
\label{eq:basic-term}
\Exp_{\alpha,\beta,\gamma,\delta} \!
\left( \tr \prod_i \rho_{\alpha i, \kappa \alpha i} \right)
\left( \tr \prod_i \rho_{\beta i, \lambda \beta i} \right) 
\left( \tr \prod_i \rho^*_{\gamma i, \mu \gamma i} \right)
\left( \tr \prod_i \rho^*_{\delta i, \nu \delta i} \right)\, .
\end{equation}
In light of Lemma~\ref{lem:fourth} (cf.~\eqref{eq:decouple}), we may ``decouple'' any four appearances of the same $\rho_{ij}$, resulting in a sum of terms in which no $\rho$ appears more than twice. For this reason, we begin our analysis with the extra assumption that each $\rho_{ij}$ appears exactly twice. For notational convenience, let us temporarily refer to the $2n$ distinct $\rho_{ij}$ appearing in~\eqref{eq:basic-term} simply by
$\rho_1, \rho_2, \ldots, \rho_{2n},$ this list in the natural order given by $\kappa$ and $\lambda$ (e.g., $\rho_i = \rho_{i, \kappa i}$ and $\rho_{n + i} = \rho_{i, \lambda i}$ for $i \leq n$). For a tuple $(\alpha, \beta, \gamma, \delta)$, then, the cupcaps of Eq.~\eqref{eq:cupcap} introduce edges between conjugate appearances of the same $\rho_i$ as shown in Figure~\ref{fig:symmetrized}; any two indices attached by an edge are constrained to be equal.

\begin{figure}[ht]
\begin{center}
\parbox[b]{8cm}{\vfill
\subfigure[Cupcaps and rotations]{%
\label{fig:symmetrized}%
\includegraphics[width=8cm]{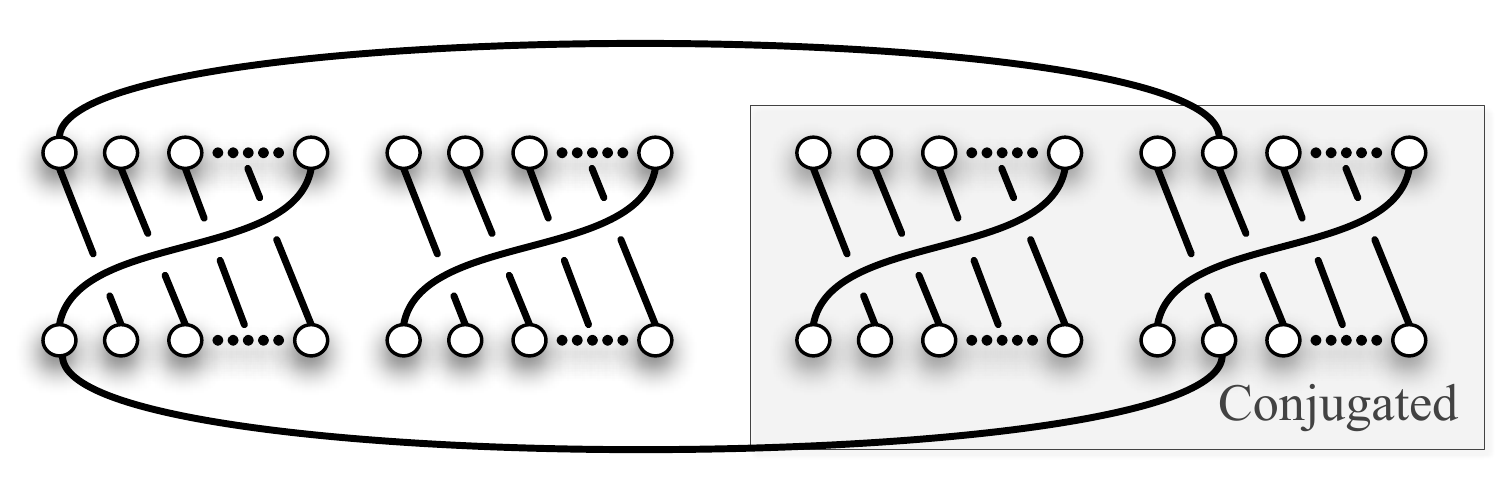}}\vfill}
\quad
\parbox[b]{8cm}{\vfill
\subfigure[Symmetrization induces conjugation]{%
\label{fig:contraction-conjugates}%
\includegraphics[width=8cm]{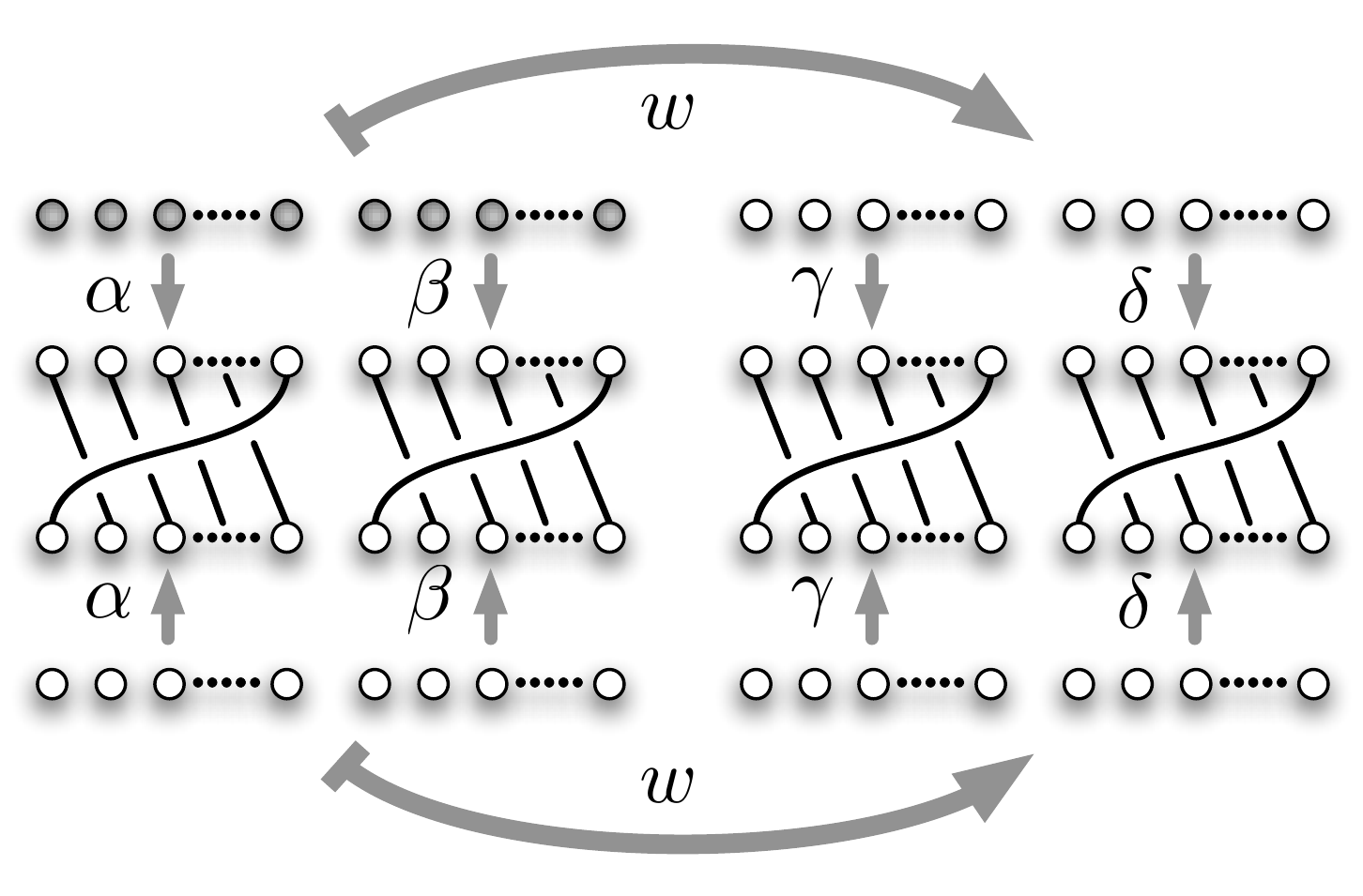}}\vfill}
\end{center}
\caption{Contractions in the second moment computation}
\label{fig:contraction}
\end{figure}

With this convention, the permutations $\mu$ and $\nu$ determine a permutation $w \in S_{2n}$ given by the ordering of the conjugate appearances of the $\rho_i$ (when $\alpha = \beta = \gamma = \delta = 1$). The contraction determined by $w$ and a particular $(\alpha, \beta, \gamma, \delta)$ is combinatorial in the sense that it merely constrains families of indices (among the $[\rho_i]^t_s$ and their conjugates) to be equal. Recalling that each cupcap contributes a factor of $1/d$ and each cycle permits $d$ different settings of the indices it contains, the value of this contraction is determined by the cycle structure of the permutation
$$
(r^{-1},r^{-1})^{(\alpha^{-1}, \beta^{-1})}\; w^{-1} \;(r, r)^{(\gamma, \delta)} \;w\, ;
$$
see Figure~\ref{fig:contraction-conjugates}. In particular, we may write the quantity of~\eqref{eq:basic-term} as
\[
\frac{1}{d^{2n}} \Exp_{\alpha, \beta, \gamma, \delta} d^{c\left( (r^{-1},r^{-1})^{(\alpha, \beta)} w^{-1} (r, r)^{(\gamma, \delta)} w \right)}\, ,
\]
where, as before, $c(\pi)$ denotes the number of cycles in the permutation $\pi$.

As we are interested in the expectation, over all rearrangements determined by $\alpha$, $\beta$, $\gamma$, and $\delta$, the only relevant feature of the permutation $w$ is
\begin{equation}
\label{eq:w_k}
k = k_{\kappa,\nu} = \Bigl| w\bigl(\{1, \ldots, n\} \cap \{ n+1, \ldots, 2n\}\bigr) \Bigr| = \Bigl|\bigl\{(i,\kappa(i))\bigr\} \cap \bigl\{(i, \nu(i))\bigr\}\Bigr|\, ,
\end{equation}
the number of $\sigma_i$ carried from the ``$\kappa$-block'' to the ``$\nu$-block.'' Defining $w_k = (1 \;n+1) \cdots (k\;n+k)$, we may rewrite the expectation of~\eqref{eq:basic-term} as
\[
\frac{1}{d^{2n}} \Exp_{\alpha, \beta, \gamma, \delta} d^{c\left( (r^{-1},r^{-1})^{(\alpha, \beta)} w_k (r,r)^{(\gamma, \delta)} w_k \right)}\, .
\]

As in Section~\ref{sec:second-unsym}, for a given double cycle cover $C$, a coloring $(\kappa, \lambda, \mu, \nu) \vdash C$ is determined by selecting, for each nontrivial cycle $c$ of $C$, whether $c$'s colors alternate between $(\kappa, \mu)$ and $(\lambda, \nu)$ or $(\kappa, \nu)$ and $(\lambda, \mu)$, and the parity of this coloring.  In light of the decoupling equation~\eqref{eq:decouple}, we may treat each isolated edge as an ``unordered pair'' of edges that can be colored in two possible ways, with $(\kappa, \mu)$ and $(\lambda, \nu)$ or $(\kappa, \nu)$ and $(\lambda, \mu)$.  Recall that in the case of Haar measure, this introduces a factor $1 + O(1/d)$ for each isolated edge, giving the factor $(1+O(1/d))^n$.

Observe now that the value of $k$ determined in Eq.~\eqref{eq:w_k} is unaffected by the choice of parity in a nontrivial cycle. The other choices described above (determining the colors involved in a nontrivial cycle or isolated edge) have the effect of exchanging a family of $\rho_{ij}$ in the $\mu$-block with a family in the $\nu$-block.

In particular, focusing on the portion of the second moment corresponding to a particular double cycle cover $C$, we have
\begin{align}
 \sum_{(\kappa,\lambda,\mu,\nu) \vdash C} 
&\, \Exp_{\alpha,\beta,\gamma,\delta} \!
\left( \tr \prod_i \rho_{\alpha i, \kappa \alpha i} \right)
\left( \tr \prod_i \rho_{\beta i, \lambda \beta i} \right) 
\left( \tr \prod_i \rho^*_{\gamma i, \mu \gamma i} \right)
\left( \tr \prod_i \rho^*_{\delta i, \nu \delta i} \right) \nonumber\\
&= \sum_{(\kappa,\lambda,\mu,\nu) \vdash C} 
\, \Exp_{\alpha,\beta,\gamma,\delta} d^{c\left( (r^{-1},r^{-1})^{(\alpha, \beta)} w_{\kappa,\nu} (r,r)^{(\gamma, \delta)} w_{\kappa,\nu} \right)} \label{eq:second-moment-w-actual} \\
&\leq \left(1 + O\!\left( \frac{1}{d} \right) \right)^n 2^{t(C)}  \sum_{k=0}^n \binom{n}{k} \frac{1}{d^{2n}} 
\Exp_{\alpha, \beta, \gamma, \delta} d^{c\left( (r^{-1},r^{-1})^{(\alpha, \beta)} w_k (r,r)^{(\gamma, \delta)} w_k \right)} \label{eq:second-moment-w-k}\\
&= \left(1 + O\!\left( \frac{1}{d} \right) \right)^{n} 2^{t(C)}  a^{(2)}_d \, .
\end{align}
Summing over all cycle covers $C \vdash A$ and applying~\eqref{eq:perm-cycle-sum} completes the proof.  For the Gaussian measure, the same proof applies without the factor $(1+O(1/d))^n$.
\end{proof}

We return to the proof of Lemma~\ref{lem:a-2-hat}.

\begin{proof}[Proof of Lemma~\ref{lem:a-2-hat}] The inequality
$$
\frac{1}{\binom{n}{n/2}} \tilde{a}_d^{(2)} \; \leq \; {a}_d^{(2)} \;\leq\; \tilde{a}_d^{(2)}
$$
is immediate from the fact that the terms of the sums defining these quantities are positive. We introduce some further notation: for a permutation $\pi \in S_{2n}$, we define 
\[
\pi^\uparrow = \bigl\{ i \mid i \in \{1, \ldots, n\}, \pi i \in \{ n+1, \ldots, 2n\} \bigr\} \quad \text{and} \quad \pi^\downarrow = \bigl\{ i \mid i \in \{ n+1, \ldots, 2n\}, \pi i \in \{ 1, \ldots, n\} \bigr\}\, .
\] 
Then $| \pi^\uparrow| = | \pi^\downarrow|$ and, if $\pi$ is selected uniformly in $S_{2n}$, 
$\Pr \bigl[|\pi^\uparrow| = k\bigr] = {\binom{n}{k}^2}/{\binom{2n}{n}}$. 
Observe also that if $\alpha$, $\beta$, $\gamma$, and $\delta$ are chosen uniformly from $S_n$, the element $(\gamma, \delta)w_k(\alpha, \beta)$ is uniform in the set $\{ \pi \mid |\pi^\uparrow| = k\}$. Recalling that $d^{c(\cdot)}$ is a class function, 
\begin{align}
\frac{1}{\binom{2n}{n}} \cdot  \tilde{a}^{(2)}_d &= \frac{1}{\binom{2n}{n}} \sum_k \binom{n}{k}^2 \frac{1}{d^{2n}} 
\Exp_{\alpha, \beta, \gamma, \delta} d^{c\left( (r^{-1},r^{-1})^{(\alpha, \beta)} w_k (r, r)^{(\gamma, \delta)} w_k \right)} 
\label{eq:symmetrized-class} \\
& =  \frac{1}{\binom{2n}{n}}  \sum_k \binom{n}{k}^2 \frac{1}{d^{2n}} 
\Exp_{\alpha, \beta, \gamma, \delta} 
d^{c\left( (r^{-1},r^{-1}) (\alpha, \beta)^{-1} w_{k} (\gamma, \delta)^{-1} (r, r) (\gamma, \delta) w_{k} (\alpha, \beta) \right)}
\nonumber\\
& =  \frac{1}{d^{2n}} \Exp_\pi d^{c\left( (r^{-1},r^{-1})  (r, r)^{\pi} \right)} 
= \frac{1}{d^{2n}}\Exp_\pi \Exp_\sigma d^{c\left( (r,r)^\sigma  (r, r)^{\pi} \right) } \, , \nonumber
\end{align}
where $\pi$ and $\sigma$ are chosen uniformly at random from $S_{2n}$.  Here we use the fact that any element of $S_{2n}$---in this case, $(r,r)$---is in the same conjugacy class as its inverse.

Defining $P_{n,n}$ to be the uniform distribution on the conjugacy class 
\[
[(r,r)] = \{ (r,r)^\pi \mid \pi \in S_{2n}\} \subset S_{2n}\,,
\]
we may express the quantity above as an inner product
\begin{equation}
\label{eq:a-2-d-inner}
\frac{1}{\binom{2n}{n}} \tilde{a}_d^{(2)} = \frac{1}{d^{2n}} (2n)! \;  \langle d^{c(\cdot)}, P_{n,n} * P_{n,n} \rangle \, .
\end{equation}

As in the proof of Lemma~\ref{lem:sd-exact}, we compute this inner product by determining the Fourier expansions of the class functions $d^{c(\cdot)}$ and $P_{n,n}$. By the Murnaghan-Nakayama rule, 
$\chi_{\lambda}(n,n)= 0$ unless the tableau $\lambda$ can be expressed as the union of two $n$-ribbon tiles. Any such tableau has \emph{rank} (the number of cells on the diagonal) no more than two and can be conveniently expressed in terms of its \emph{characteristics}: defining $a_i$ and $b_i$ to be the number of cells below and to the right of the $i$th box of the diagonal, respectively, we use the notation $\tau = (b_1, b_2, \ldots, b_r \mid a_1, a_2, \ldots, b_r)$ to describe the tableau (see Figure~\ref{fig:ribbons}).
\begin{figure}[ht]
\centering
\includegraphics[width=7cm]{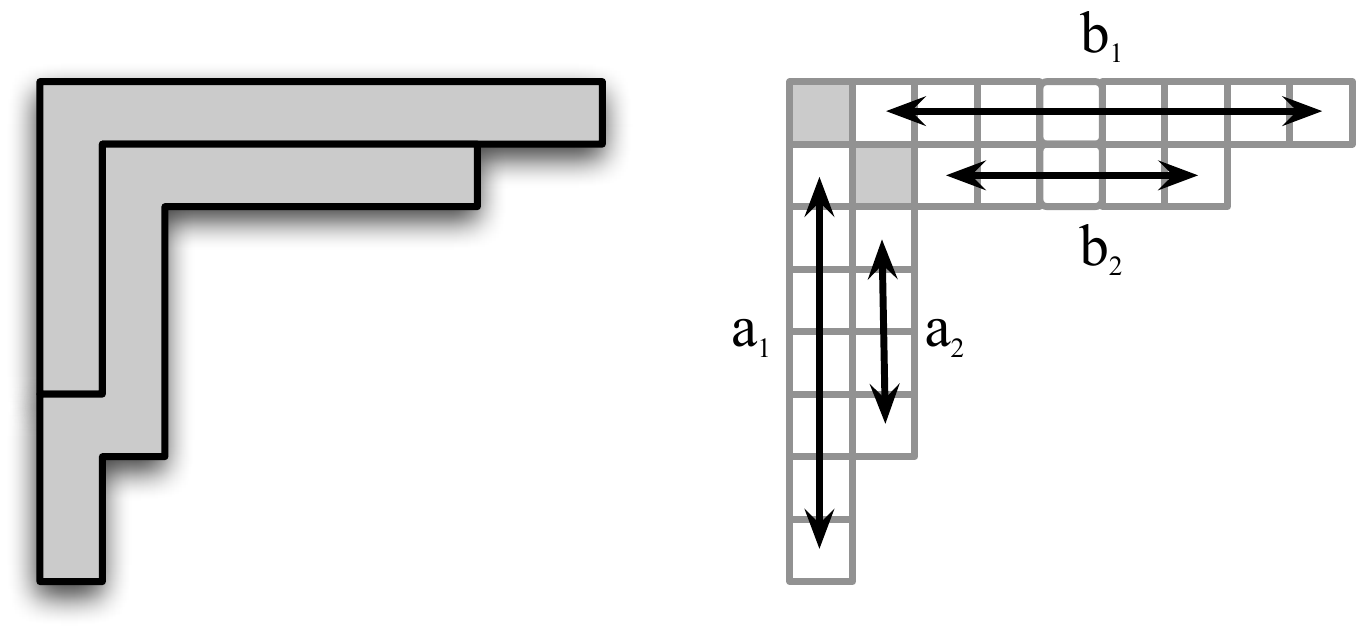}
\caption{A Young tableau decomposed into two $n$-ribbon tiles.}
\label{fig:ribbons}
\end{figure}
If $\chi_\tau(n,n)$ is nonzero, so that $\tau$ can be written as the union of two $n$-ribbons, we find (again appealing to the Murnaghan-Nakayama rule) that either
\begin{itemize}
\item $\tau = (b_1, b_2 \mid a_1, a_2)$ has rank two, $a_1 + b_2 + 1 = a_2 + b_1 + 1 = n$, and $\chi_\tau(n,n) = \pm 2$, or
\item $\tau = (b_1 \mid a_1)$ has rank one and $\chi_\tau(n,n) = \pm 1$.
\end{itemize}
We let $T_n$ denote the family of representations of $S_{2n}$ described above; note that $\abs{T_n} \le n^2$. Observe that for each $\tau \in T_n$, $\langle P_{n,n}, \chi_\tau\rangle = \frac{1}{(2n)!} \chi_\tau(n,n)$ (where $\chi_\tau(n,n) \in \{ \pm 1, \pm 2\}$). Recalling that
$$
\chi * \chi = \frac{|G|}{\chi(1)} \chi
$$
for any irreducible character $\chi$ of a group $G$, we may express
$$
\langle P_{n,n} * P_{n,n}, \chi_\tau\rangle = \frac{1}{(2n)!} \frac{\chi_\tau(n,n)^2}{\dim \tau}\,.
$$
As discussed in the proof of Lemma~\ref{lem:sd-exact},
$$
\inner{d^{c(\cdot)}}{\chi_\tau} = \inner{\chi_\Sigma}{\chi_\tau} = \sum_{\substack{(\rho_1, \ldots, \rho_d)\\\sum \rho_i = 2n}} K^\tau_\rho\, ,
$$
where $\chi_\Sigma$ is the permutation representation given by the action of $S_{2n}$ on the set $\{ (a_1, \ldots, a_{2n} \mid a_i \in \{1, \ldots, d\}\}$ and $K^\tau_\rho$ is the Kostka number, equal to the number of semistandard tableaux of shape $\tau$ with $\rho_i$ appearances of the number $i$. Then
\begin{equation}
\label{eq:a-2-d-inner-exact}
\langle P_{n,n} * P_{n,n}, d^{c(\cdot)}\rangle 
= \sum_\tau \langle P_{n,n} * P_{n,n}, \chi_\tau\rangle \inner{d^{c(\cdot)}}{\chi_\tau} 
= \frac{1}{(2n)!} \sum_{\tau \in T_n} \chi_\tau(n,n)^2 \sum_{\substack{(\rho_1, \ldots, \rho_d) \\ \sum \rho_i = 2n}} \frac{K^\tau_\rho}{\dim \tau} \, .
\end{equation}
Note that for each $\tau \in T_n$, $\chi_\tau(n,n)^2 \leq 4$ and $K^\tau_\rho \leq \dim \tau$, as $\dim \tau$ is the number of semistandard tableaux of shape $\tau$ with distinct entries in any totally ordered set. Thus,  
\begin{equation*}
\langle P_{n,n} * P_{n,n}, d^{c(\cdot)}\rangle 
\leq \frac{4}{(2n)!} \sum_{\tau \in T_n} \binom{2n + d - 1}{2n} 
\leq \frac{4}{(2n)!} \abs{T_n} \binom{2n + d - 1}{2n} 
\le \frac{4n^2}{(2n)!} \binom{2n + d - 1}{2n} \, .
\end{equation*}
On the other hand, each term in the sum of~\eqref{eq:a-2-d-inner-exact} is positive; thus
\begin{equation}
\label{eq:pnn-upper}
\langle P_{n,n} * P_{n,n}, d^{c(\cdot)}\rangle 
\geq \langle P_{n,n} * P_{n,n}, \chi_1\rangle \inner{d^{c(\cdot)}}{\chi_1} = \frac{1}{2n!} \binom{2n + d - 1}{2n}\,.
\end{equation}
We conclude that
$$
\frac{1}{2n!} \binom{2n + d - 1}{2n} \leq \langle P_{n,n} * P_{n,n}, d^{c(\cdot)}\rangle 
 \le \frac{4n^2}{2n!} \binom{2n + d - 1}{2n} 
$$
which, in conjunction with~\eqref{eq:a-2-d-inner}, completes the proof of Lemma~\ref{lem:a-2-hat}.
\end{proof}

Now we apply these Lemmas to prove an upper bound on the critical ratio $\Exp[\Xsym^2]/\Exp[\Xsym]^2$.  If $d$ is constant, which is the only case for which we have an efficient algorithm to compute $\Xsym$~\cite{barvinok-sdet}, our bound is not very inspiring.  If $n \ge d$, combining Lemmas~\ref{lem:sd-exact}, \ref{lem:symmetrized-second-diagram}, and \ref{lem:a-2-hat} gives
\begin{equation*}
\frac{\Exp[\Xsym^2]}{\Exp[\Xsym]^2} 
\le 4n^2 \binom{2n}{n} \binom{2n + d - 1}{2n} \Big\slash {n+d \choose n+1}^2 
= O(n^3/d) {2n+2d \choose n+d} \Big\slash {2n+2d \choose d} 
= 2^{2n} \,n^{-d+O(1)} \, ,
\end{equation*}
assuming that $d=O(1)$.  This proves~\eqref{eq:sym-constant} in Theorem~\ref{thm:main}, and suggests that $d$ needs to grow with $n$ to give a good estimator.  

On the other hand, when $d$ grows fast enough with $n$, we find that the critical ratio behaves quite well.  Combining Lemma~\ref{lem:a-2-hat} with the lower bound~\eqref{eq:sd-lower} gives
\begin{equation*}
\frac{\Exp[\Xsym^2]}{\Exp[\Xsym]^2} 
\le \frac{4 n!^2}{d^{2n}} \binom{2n}{n} \binom{2n + d - 1}{2n} 
= \frac{4}{d^{2n}} \frac{(2n + d - 1)!}{(d-1)!} 
= 4 \left(1+\frac{1}{d} \right) 
\cdots \left(1+\frac{2n-1}{d} \right) 
\le 4 \e^{4n^2/d} \, ,
\end{equation*}
completing the proof of~\eqref{eq:sym-growing} in Theorem~\ref{thm:main}.  

In the critical case where $d = O(1)$ the upper bound of $2^{2n}n^{-d+O(1)}$ we establish above is tight up to the factor introduced by our ``approximation'' of $a_d^{(2)}$ by $\tilde{a}_d^{(d)}$---that is, a factor of $\binom{n}{n/2}$. In particular, even for the identity matrix, we can establish a $2^n n^{-d+O(1)}$ lower bound on the critical ratio:

\begin{theorem*}[Restatement of Theorem~\ref{thm:lower}]
Let $A$ be the $n \times n$ identity matrix and $d$ a constant.  Then
\begin{equation*}
\frac{\Exp[\Xsym^2]}{\Exp[\Xsym]^2} = \Omega\left(\frac{2^n}{n^d}\right) 
\quad \text{and} \quad 
\frac{\Exp[\Xsym^2]}{\Exp[\Xsym]^2} = \left(1 - O\!\left(\frac{1}{d} \right)\right)^n \Omega\left(\frac{2^n}{n^d}\right)\,,
\end{equation*}
when the $\rho_{ij}$ are distributed according to the Gaussian or Haar measure respectively.
\end{theorem*}

\begin{proof}[Proof of Theorem~\ref{thm:lower}]
Let $d$ be a constant and $A$ the $n \times n$ identity matrix. Then $\perm A= \perm^2 A = 1$ and, from Lemma~\ref{lem:sd-exact},
$$
\Exp[\Xsym] = a_d = \frac{1}{d^n} \binom{n+d}{n+1}\,.
$$
As for the second moment, the only nontrivial term in the sum~\eqref{eq:second-moment} corresponds to the case where the permutations $\kappa$, $\lambda$, $\mu$, and $\nu$ are the identity. In this case each $\rho_{ij}$ appears four times and there are precisely $\binom{n}{k}$ terms of~\eqref{eq:second-moment-w-actual} for which $w_{\kappa,\lambda} = w_k$; in particular, in this case the inequality of~\eqref{eq:second-moment-w-k} is an equality. Recalling Lemma~\ref{lem:fourth}, we conclude that
$$
\Exp[\Xsym^2] = a_d^{(2)} \quad \text{and} \quad \Exp[\Xsym^2] \geq (1 - O(1/d))^n\; a_d^{(2)}
$$
when the $\rho_{ij}$ have Gaussian measure and Haar measure, respectively. For constant $d$ we have
$$
\frac{\tilde{a}_{d}^{(2)}}{a_d^2} \geq \frac{{a}_{d}^{(2)}}{\binom{n}{n/2} a_d^2} = \frac{\binom{2n}{n}\binom{2n+d-1}{2n}}{\binom{n}{n/2} \binom{n+d}{n+1}^2}
$$
and, considering that $\binom{\ell}{\ell/2} = \frac{2^\ell}{\Theta(\sqrt{\ell})}$ and $\binom{2n + d - 1}{2n} \geq \binom{n+d}{n+1}$, 
$$
\frac{\tilde{a}_d^{(2)}}{a_d^2} = \frac{2^{2n}}{2^n O(\sqrt{n}) \binom{n+d}{n+1}} = \Omega\left(\frac{2^n}{n^d}\right)\,.
$$
The statement of the theorem follows.
\end{proof}

\section{Estimators based on the Frobenius norm}
\label{sec:frob}

In this section, we prove Theorem~\ref{thm:frob} by relating the moments of Frobenius estimators, $\Xfrobunsym=\norm{\det M}^2$ and $\Xfrobsym=\norm{\sdet M}^2$, to those of the trace-squared estimators we studied above.

As Fig.~\ref{fig:frob} shows, the diagrams corresponding to the expectations and second moments of these estimators differ from those of their counterparts by a small number of local moves.  Let $Q$ be the product of some sequence of $\rho_{ij}$.  Then all we have to do is change our previous contraction,
\[
\abs{\tr Q}^2 = Q^i_i Q^j_j 
\]
where the ``output'' of each $Q$ is connected to its ``input,'' to
\[
\norm{Q}^2 = \tr Q Q^\dagger = Q^i_j (Q^\dagger)^j_i = Q^i_j (Q^*)^i_j \, . 
\]
In this contraction, we connect the output of each $Q$ to the output of the corresponding $Q^*$, and similarly wire their inputs together.  The cupcaps, resulting from taking the expectation of $\rho \otimes \rho^*$ for each $\rho_{ij}$ appearing in these products, remain the same as before.

\begin{figure}
\begin{center}
\includegraphics[width=10cm]{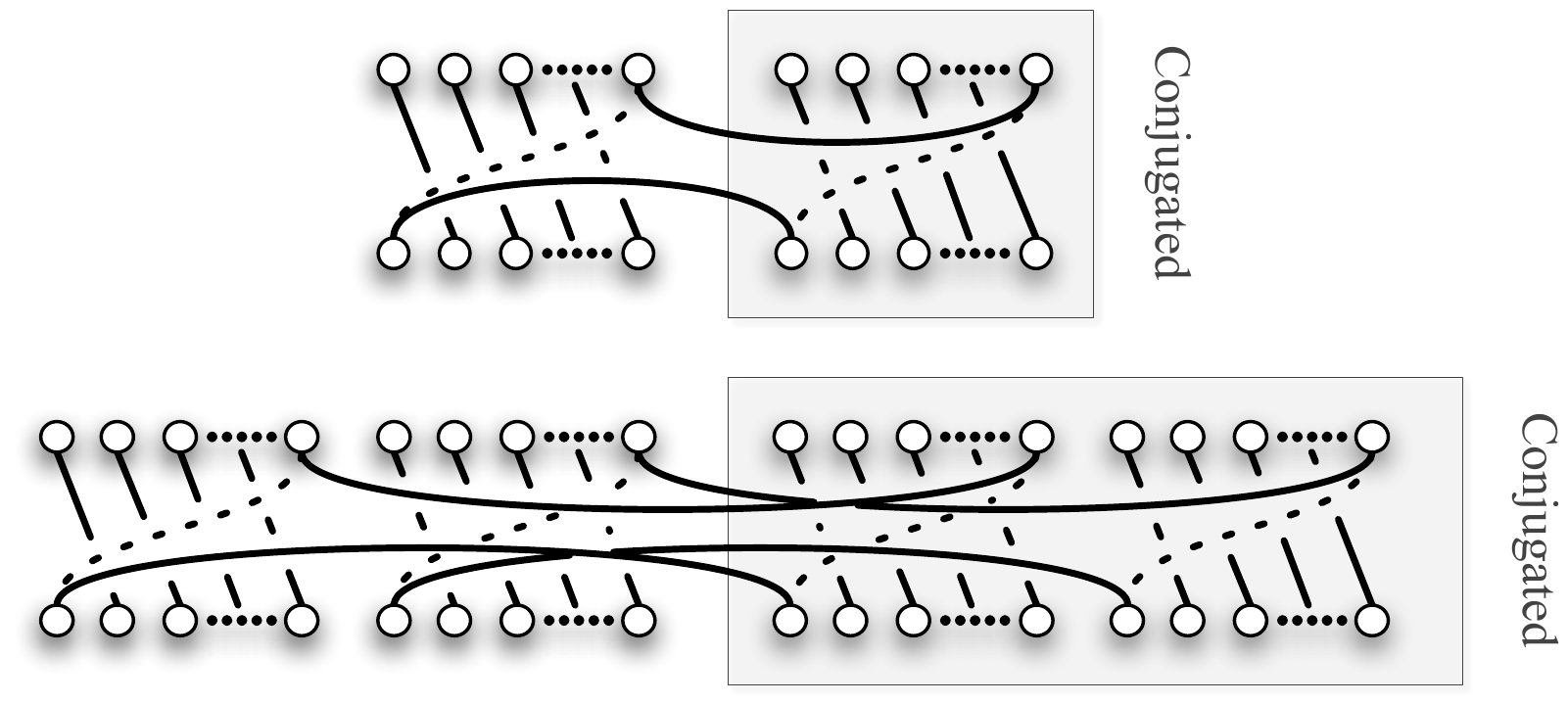}
\end{center}
\caption{Rewiring the diagram to change $\abs{\tr M}^2$ to $\tr M M^\dagger = \norm{M}^2$.  The cupcaps remain unchanged, but instead of wiring the ``input'' of each product to its ``output,'' we wire a pair of products together ``input'' to ``input'' and ``output'' to ``output.''}
\label{fig:frob}
\end{figure}

Now recall that the expectation and second moment of these estimators is proportional to $d^c$, where $c$ is the number of loops in these diagrams.  Each of these rewiring moves changes the number of loops by at most one, by cutting one loop into two or merging two loops into one.  Thus we have 
\[
\frac{1}{d} \Exp[\Xunsym] \le \Exp[\Xfrobunsym] \le d \,\Exp[\Xunsym]
\quad \text{and} \quad
\frac{1}{d^2} \Exp[\Xunsym^2] \le \Exp[\Xfrobunsym^2] \le d^2 \,\Exp[\Xunsym^2] \, ,
\]
and similarly in the symmetrized case.  Assuming the worst regarding these bounds yields~\eqref{eq:frob}, and completes the proof of Theorem~\ref{thm:main}.

\section{Conclusions}
\label{sec:conclusion}

As we stated in the Introduction, our results present us with the following irony.  For the estimators based on the unsymmetrized determinant, which we do not know how to compute efficiently, the critical ratio $\Exp[X^2]/\Exp[X]^2$ becomes more mildly exponential as $d$ increases.  Specifically, for any $\eps > 0$ we can make the critical ratio $O((1+\eps)^n)$ by taking $d=1/\eps$.  

On the other hand, for the estimators based on the symmetrized determinant, the critical ratio is $\Omega(2^n)$ in the case $d=O(1)$ where we have an efficient algorithm.  In order to reduce this exponential to $O(c^n)$ for some $c < 2$, we need $d$ to be a growing function of $n$.  This is contrary to the intuition expressed in~\cite{barvinok-sdet}, and to our own initial intuition when we began work on this problem.  

\begin{sloppypar}
Of course, the symmetrized estimators may still be tightly concentrated, as conjectured in~\cite{barvinok-sdet}.  However, since their variance is large, any proof of concentration would have to bound, implicitly or explicitly, their higher moments.
\end{sloppypar}

At this point, finding an algebraic polynomial-time approximation scheme for the permanent seems to require progress on at least one of several fronts.  One approach would be to seek a polynomial-time algorithm for $\sdet M$ in the case where $M$'s entries belong to $\alg_d$ where $d = \poly(n)$, but it seems difficult to scale up the algorithm of~\cite{barvinok-sdet} beyond $d=O(1)$.  Another approach, as suggested in~\cite{chien-rasmussen-sinclair}, would be to seek an algorithm for $\det M$ where $M$'s entries belong to some group with representations of arbitrarily high dimension.  However, it seems difficult to construct a succinct description for the group algebra elements which appear in the determinant, since their support in the group basis is exponentially large.  

\section*{Acknowledgments}

This research was supported by NSF grants CCF-0524613, CCF-0835735, and CCF-0829917.

\newcommand{\etalchar}[1]{$^{#1}$}

\appendix

\section{Representation theory and the symmetric group}
\label{sec:representation-theory}

We briefly discuss the elements of the representation theory of groups, and of the symmetric groups in particular. Our treatment is primarily for the purposes of setting down notation; we refer the reader to~\cite{JamesK1981} for a complete account.

Let $G$ be a finite group. A \emph{representation} $\rho$ of $G$ is a
homomorphism $\rho: G \to \U(V)$, where $V$ is a finite-dimensional
Hilbert space and $\U(V)$ is the group of unitary operators on $V$.
The \emph{dimension} of $\rho$, denoted $d_\rho$, is the dimension of the
vector space $V$.  By choosing a basis for $V$, then, we can identify each $\rho(g)$ 
with a unitary $d_\rho \times d_\rho$ matrix; these matrices then satisfy $\rho(gh) = \rho(g) \cdot \rho(h)$  for every $g, h \in G$.

Fixing a representation $\rho: G \to \U(V)$, we say that a subspace $W \subset
V$ is \emph{invariant} if $\rho(g) W \subset W$ for all $g \in G$.
We say $\rho$ is \emph{irreducible} if it has no invariant subspaces other than 
the trivial space $\{ \vec{0} \}$ and $V$.  
If two representations $\rho$ and $\sigma$ are the same up to a unitary
change of basis, we say that they are \emph{equivalent}. It is a fact
that any finite group $G$ has a finite number of distinct irreducible
representations up to equivalence and, for a group $G$, we let
$\hat{G}$ denote a set of representations containing exactly one
from each equivalence class. The irreducible representations of $G$
give rise to the Fourier transform.  Specifically, for a function $f: G
\to \C$ and an element $\rho \in \hat{G}$, define the \emph{Fourier
  transform of $f$ at $\rho$} to be
$$
\hat{f}(\rho) = \sqrt{\frac{d_\rho}{|G|}} \sum_{g \in G} f(g)\rho(g)\enspace.
$$
The leading coefficients are chosen to make the transform unitary,
so that it preserves inner products:
$$
\langle f_1, f_2\rangle = \sum_g f_1^*(g)f_2(g) 
= \sum_{\rho \in \hat{G}}\tr \!\left(\hat{f_1}(\rho)^\dagger \cdot \hat{f_2}(\rho)\right)\enspace.
$$

In the case when $\rho$ is \emph{not} irreducible, 
it can be decomposed into a direct sum of irreducible representations, each one of which 
operates on an invariant subspace.  We write $\rho = \sigma_1 \oplus \cdots \oplus \sigma_k$ 
and, for the $\sigma_i$ appearing at least once in this decomposition, $\sigma_i \prec \rho$.
In general, a given $\sigma$ can appear multiple times, 
in the sense that $\rho$ can have an invariant subspace isomorphic to
the direct sum of $a^\rho_\sigma$ copies of $\sigma$.  In this case $a^\rho_\sigma$ is
called the \emph{multiplicity} of $\sigma$ in $\rho$, and we write $\rho =
\bigoplus_{\sigma \prec \rho} a^\rho_\sigma \sigma$.

For a representation $\rho$ we define its \emph{character} as the trace
$\chi_\rho(g) = \tr \rho(g)$.  
Given an element $m$, we denote its conjugacy class $[m]=\{ g^{-1} m g \mid g \in G\}$.  
Since the trace is invariant
under conjugation, characters are constant on the conjugacy classes, and we write
$\chi_\rho([m]) = \chi_\rho(m)$ where $m$ is any element of $[m]$. 
Characters are a powerful tool for reasoning about
the decomposition of reducible representations.  In particular, for
$\rho, \sigma \in \hat{G }$, we have the orthogonality conditions
$$
\langle \chi_\rho, \chi_\sigma \rangle_{G} = \frac{1}{|G|} \sum_{g \in G} \chi_\rho(g)\chi_\sigma(g)^*
= \begin{cases} 1 & \rho = \sigma\enspace,\\
  0 & \rho \neq \sigma\enspace. \end{cases}
$$
If $\rho$ is reducible, we have $\chi_\rho = \sum_{\sigma \prec \rho} a^\rho_\sigma \chi_{\sigma_i}$, and so the multiplicity $a^\rho_\sigma$ is given by 
\[ a^\rho_\sigma = \langle \chi_\rho, \chi_\sigma \rangle_{G} \enspace . \]

If $\rho$ is irreducible, \emph{Schur's lemma} asserts that the only matrices which commute 
with $\rho(g)$ for all $g$ are the scalars, $\{ c \one \mid c \in \C \}$.  Therefore, for any $A$ 
we have
\begin{equation}
\label{eq:a}
 \frac{1}{|G|} \sum_{g \in G} \rho(g)^\dagger A \rho(g) = \frac{\tr A}{d_\rho} \one_{d_\rho}  
\end{equation}
since conjugating this sum by $\rho(g)$ simply permutes its terms.

We specialize now to the case of the symmetric group $S_n$ of permutations of the set $\{ 1, \ldots, n\}$. The representations of $S_n$ are in one-to-one correspondence with \emph{Young diagrams} or, equivalently, integer partitions $\lambda = (\lambda_1, \lambda_2, \cdots)$ where $\lambda_1 \geq \lambda_2 \geq \cdots$ and $\sum_i \lambda_i = n$. The character of this representation is denoted $\chi_\lambda$. The \emph{Murnaghan-Nakayama rule} gives a recursive formula for the character $\chi_\lambda$. In preparation for stating the rule, we define a \emph{ribbon tile} of length $k$ to be a polyomino of $k$ cells, arranged in a path where each step is up or to the right.

\begin{lemma}[Murnaghan-Nakayama rule]
\label{lem:mn}
 Given a Young diagram $\lambda$ and a permutation $\pi$ with cycle structure $k_1 \ge k_2 \ge \cdots$, a \emph{consistent tiling} of $\lambda$ consists of removing a ribbon tile of length $k_1$ from the boundary of $\lambda$, then one of length $k_2$, and so on, with the requirement that the remaining part of $\lambda$ is a Young diagram at each step.  Let $h_i$ denote the height of the ribbon tile corresponding to the $i$th cycle: then
\begin{equation}
\label{eq:m-n}
 \chi_\lambda(\pi) = \sum_T \prod_i (-1)^{h_i+1} 
\end{equation}
where the sum is over all consistent tilings $T$.
\end{lemma}

\section{Proof of Lemma~\ref{lem:cupcap}}
\label{sec:cupcap}

\begin{proof}
For the Gaussian measure, this is simply the fact that 
$\left( \sigma \otimes \sigma^* \right)^{ik}_{j\ell} 
= \sigma^i_j (\sigma^k_\ell)^*$.  
If $i \ne k$ or $j \ne \ell$, then this is the product of two independent random variables both of whom have expectation zero.  If $i=k$ and $j=\ell$, then this is $\abs{ \sigma^i_j }^2$, whose expectation is $1/d$.

For the Haar measure, \eqref{eq:cupcap-coord} follows from a little representation theory.  (For a brief introduction to representation theory, see Appendix~\ref{sec:representation-theory}.)  Abusing notation, suppose that $\sigma$ is the defining representation of the group $\U(d)$ of unitary matrices, i.e., the $d$-dimensional representation in which unitary matrices act on column vectors in the natural way.  Then $\sigma \otimes \sigma^*$ is isomorphic to the conjugation action of $\U(d)$ on $\GL(d)$, the vector space of $d \times d$ matrices.  We can decompose this into the direct sum of two invariant subspaces 
$\sigma \otimes \sigma^* \cong \one \oplus \Gamma$, 
where $\one$ is the trivial representation, consisting of the scalar matrices, and $\Gamma$ is the $(d^2-1)$-dimensional representation consisting of $d \times d$ matrices with zero trace.  Both these subspaces are clearly invariant under conjugation, and are, in fact, irreducible.  
Taking the expectation over $\sigma \in \U(d)$ gives the projection operator $\Pi^{\sigma \otimes \sigma^*}_\one$ onto the trivial subspace---that is, the linear operator on the space of matrices which takes a matrix $A = A_{ik}$ and returns a scalar whose trace is $\tr A$.  We claim that this operator is exactly~\eqref{eq:cupcap-coord}, since 
\[
\left( \frac{1}{d} \,\delta^{ik} \delta_{j\ell} \right) A_{ik} = \frac{1}{d} \,A^i_i \delta_{j\ell} = \left( \frac{1}{d} \,\tr A \right) \one \, . 
\]
Here we again use the Einstein summation convention, so that $A^i_i = \tr A$, and the identity matrix is $\one=\delta_{j\ell}$.
\end{proof}

\end{document}